\documentclass[10pt, a4paper, twocolumn]{IEEEtran}

\usepackage{amsmath,epsfig,amssymb,verbatim,amsopn,cite,multirow}
\usepackage{amsthm}
\usepackage{balance}
\usepackage{multirow}
\usepackage[usenames,dvipsnames]{color}
\usepackage[all]{xy}  %%% used to make block diagram
\usepackage{url}
\usepackage{amsfonts}
\usepackage{amssymb}
\usepackage{epsfig}
\usepackage{epstopdf}
\usepackage{bm}
\usepackage{balance}
\usepackage{graphicx}
\usepackage{subcaption}
\usepackage{footnote}
\usepackage{cancel}
\usepackage{algorithm,algorithmic}

\newtheorem{proposition}{Proposition}

\newcounter{mytempeqcounter}

\newcommand{\qe}{{\bf e}}

\newcommand{\qg}{{\bf g}}
\newcommand{\qh}{{\bf h}}

\newcommand{\qn}{{\bf n}}

\newcommand{\qw}{{\bf w}}
\newcommand{\qx}{{\bf x}}
\newcommand{\qy}{{\bf y}}
\newcommand{\qz}{{\bf z}}

\newcommand{\qB}{{\bf B}}

\newcommand{\qH}{{\bf H}}
\newcommand{\qI}{{\bf I}}

\newcommand{\PMRT}{\mathtt{PMRT}}
\newcommand{\PZF}{\mathtt{PZF}}

\newcommand{\SINRE}{\mathrm{SINR}_{E}}
\newcommand{\SINRone}{\mathrm{SINR}_{1}}
\newcommand{\SINRk}{\mathrm{SINR}_{k}}

\newcommand{\Mone}{\mathcal{M}_{1}}
\newcommand{\Zone}{\mathcal{Z}_{1}}
\newcommand{\Zt}{\mathcal{Z}_{t}}
\newcommand{\Mt}{\mathcal{M}_{t}}

\newcommand{\Betalk}{\beta_{l,k}}
\newcommand{\betalE}{\beta_{l,E}}

\newcommand{\BUE}{\mathrm{BU}_{E,1}} 
\newcommand{\UIEt}{\mathrm{UI}_{E,t}}

\newcommand{\BphiE}{\boldsymbol{\phi}_E}
\newcommand{\Bphik}{\boldsymbol{\phi}_k}
\newcommand{\Bphikp}{\boldsymbol{\phi}_{k'}}

\newcommand{\Bphione}{\boldsymbol{\phi}_1}
\newcommand{\Bvarepsilon}{\boldsymbol{\pi}}

\newcommand{\betalk}{\beta_{l,k}}

\newcommand{\betalone}{\beta_{l,1}}
\newcommand{\hlmk}{\qh_{l,k}}
\newcommand{\glmk}{\qg_{l,k}}
\newcommand{\hlmE}{\qh_{l,E}}
\newcommand{\glmE}{\qg_{l,E}}

\newcommand{\hhatlmk}{\hat {\qh}_{l,k}}

\newcommand{\gamalmk}{\gamma_{l, k}}
\newcommand{\gamalmE}{\gamma_{l, E}}
\newcommand{\gamapE}{\gamma_{p,E}}

\newcommand{\rhopk}{\rho_{p, k}}
\newcommand{\rholk}{\rho_{l, k}}
\newcommand{\rhopt}{\rho_{p, t}}
\newcommand{\rholt}{\rho_{l, t}}

\newcommand{\tausl}{|\mathcal{S}_{l}|}
\newcommand{\tausp}{|\mathcal{S}_{p}|}
\newcommand{\tausq}{|\mathcal{S}_{q}|}
\newcommand{\Sl}{\mathcal{S}_{l}}

\newcommand{\Ex}{\mathbb{E}}

\newcommand{\trace}{\mathrm{tr}}

\newtheorem{remark}{Remark}

\title{Secure Transmission in Cell-Free Massive MIMO under Active Eavesdropping}

\author{Yasseen Sadoon Atiya, Zahra Mobini,~\IEEEmembership{Member,~IEEE,} Hien Quoc Ngo,~\IEEEmembership{Senior Member,~IEEE,}\\ and Michail Matthaiou,~\IEEEmembership{Fellow,~IEEE}}

\allowdisplaybreaks

\begin{document}
	
\bstctlcite{IEEEexample:BSTcontrol}
\maketitle
\begin{abstract}
% We study  secure communications in cell-free massive multiple-input multiple-output (CF-mMIMO) systems with multi-antenna access points (APs) and  protective partial zero-forcing (PPZF) precoding. In particular, we consider an active eavesdropping attack,  where an eavesdropper  contaminates the uplink channel estimation phase by sending an identical pilot sequence with a legitimate user of interest. We formulate an optimization problem for maximizing the received signal-to-noise ratio (SINR) at the legitimate user, subject to a maximum allowable SINR at the eavesdropper and maximum transmit power at each AP, while guaranteeing specific SINR requirements on other legitimate users. The optimization problem is solved using a path-following algorithm.  We also propose a large-scale-based greedy AP selection scheme to improve the ${\color{red}{[R_3C_2]}}$ \comm{secrecy spectral efficiency (SSE)}. Finally, we propose a  method for identifying the presence of an eavesdropper within the system. Our findings show that PPZF can substantially outperform the conventional maximum-ratio transmission (MRT)  scheme by providing around $2$-fold improvement in the SSE compared to the MRT scheme. More importantly, for PPZF precoding scheme,  our proposed AP selection can achieve a remarkable SSE gain of up to $220\%$, while our power optimization approach  can provide an  additional gain  of up to   $55\%$ compared with a CF-mMIMO system with equal power allocation.

We study  secure communications in cell-free massive multiple-input multiple-output (CF-mMIMO) systems with multi-antenna access points (APs) and  protective partial zero-forcing (PPZF) precoding. In particular, we consider an active eavesdropping attack,  where an eavesdropper  contaminates the uplink channel estimation phase by sending an identical pilot sequence with a legitimate user of interest.  
We formulate an optimization problem for maximizing the received signal-to-noise ratio (SINR) at the legitimate user, subject to a 
maximum allowable SINR at the eavesdropper and maximum 
transmit power at each AP, while guaranteeing specific 
SINR requirements on other legitimate users. The optimization problem is solved using a path-following algorithm.  We also propose a large-scale-based greedy AP selection scheme 
to improve the secrecy spectral efficiency (SSE). Finally, we propose a simple method for identifying the presence of an eavesdropper within the system. Our findings show that PPZF can substantially outperform the conventional maximum-ratio transmission (MRT)  scheme by providing around $2$-fold improvement in the SSE  compared to the MRT scheme. 
More importantly, for PPZF precoding scheme,  our proposed AP selection can achieve a remarkable SSE gain of up to $220\%$, while our power optimization approach  can provide an  additional gain  of up to   $55\%$ compared with a CF-mMIMO system with equal power allocation.

\let\thefootnote\relax\footnotetext{The authors are with the Centre for Wireless Innovation (CWI), Queen's University Belfast, BT3 9DT Belfast, U.K. email:\{yhimiari01, zahra.mobini, hien.ngo, m.matthaiou\}@qub.ac.uk. Yasseen Sadoon Atiya is also a lecturer at Imam Alkadhim University College.

This work is a contribution by Project REASON, a UK Government funded project under the Future Open Networks Research Challenge (FONRC) sponsored by the Department of Science Innovation and Technology (DSIT). It was also supported by the U.K. Engineering and Physical Sciences Research Council (EPSRC) (grant
No. EP/X04047X/1).
The work of H. Q. Ngo was supported by the U.K. Research and Innovation Future Leaders Fellowships under Grant MR/X010635/1. The work of M. Matthaiou has received funding from the European Research Council (ERC) under the European Union’s Horizon 2020 research and innovation programme (grant agreement No. 101001331). Parts of this paper were presented at IEEE VTC2023-Spring~\cite{YASSEEN:2023:VTC}.}
\end{abstract}

\begin{IEEEkeywords}
Physical layer security, access point selection, active eavesdropping, cell-free massive multiple-input multiple-output, power control, secrecy.
\end{IEEEkeywords}

%-------------------------------------
% \section{Introduction} 
Cell-free massive multiple-input multiple-output (CF-mMIMO) has been envisaged as one of the promising technologies for the next generation wireless networks. By breaking the concept of cell boundaries, deploying a large number of  geographically distributed access points (APs), and  coherently serving users in the same time-frequency
resources, it avails of all the benefits granted by the state of the art techniques including massive MIMO, distributed antenna systems, and coordinated multipoint with joint transmission. CF-mMIMO  systems bring the
APs geographically closer to the users and, hence, can provide seamless and handover free services for all users by exploiting macro-diversity gains and low path losses\cite{Matthaiou:COMMag:2021,interdonato2019ubiquitous,Zhang:JSAC:2020}.

While this distributed underpinning network architecture results in ubiquitous coverage at high spectral efficiency (SE), it also escalates the vulnerability of CF-mMIMO to malicious eavesdroppers, especially when the number of APs and users grows~\cite{zahra:2023:globecm}. 
More precisely, since the APs are densely distributed over the area of coverage, the distances between the APs and users or the potential eavesdroppers are shortened, which can increase the risk of confidential information leakage.  Therefore, the security of CF-mMIMO against eavesdropping and cyber-physical attacks is of great practical significance. 
Eavesdropping is typically classified into two
main paradigms~\cite{Kapetanovic:COMMAg:2015}: 1) passive eavesdropping and 2) active eavesdropping. In passive eavesdropping, eavesdroppers  silently overhear the information delivery between APs and the targeted legitimate users without sending any pilot or interference signals~\cite{Kapetanovic:COMMAg:2015,Zahra:TWC:2019}, while in the active eavesdropping, the active eavesdroppers intervene in the communication by either sending jamming signals and/or sending spoofing pilot sequences~\cite{Zhou:TWC:2012}. In a pilot spoofing  attack, the uplink pilot training phase of the  users of interest will be attacked by  active eavesdroppers. More specifically,  because the pilot sequences are publicly available and follow standardization, malicious eavesdroppers have the ability to actively transmit spoofing pilot sequences which results in a pilot contamination attack and, hence, information leakage. It has been shown that the adverse  effects of active eavesdropping attacks are much more detrimental than those of passive attacks~\cite{Kapetanovic:COMMAg:2015}.

It is now worth noting the substantial amount of research interest that has been sparked in recent years in the implementation of physical-layer security techniques in massive MIMO systems. In particular, various methods for active pilot spoofing attack detection  have been proposed in~\cite{Xiong:TIFS:2015,Xie:ICC:2017,Zhang:Acess:2019,Li:systemJ:2023}. Additionally,  the authors in~\cite{Fan:TCOM:2019,zhu:TWC:2014,Wu:TIFS:2016,Nguyen:JSAC:2018,Xu:TVT:2021,Chen:TIFS:2016, Guo:TWC:2016,Li:TVT:2018,Zhu:JSAC:2018,Lin:WCL:2021} sought to enhance the security of massive MIMO systems
either by using cooperative jamming~\cite{Fan:TCOM:2019} and artificial noise~\cite{zhu:TWC:2014,Wu:TIFS:2016,Nguyen:JSAC:2018,Xu:TVT:2021}  for degrading the eavesdropping  rate or by using resource allocation techniques~\cite{Chen:TIFS:2016, Guo:TWC:2016} and  beamforming designs~\cite{Li:TVT:2018,Zhu:JSAC:2018,Lin:WCL:2021}  for strengthening the legitimate links. However, in the context of secure CF-mMIMO systems, there have been only a few recent works~\cite{Timilsina:GC:2018, Hoang:TCOM:2018,Zhang:SYS.2020,Alageli:TIFS:2020}. In particular,  Timilsina \emph{et al.}~\cite{Timilsina:GC:2018} derived  the secrecy spectral efficiency (SSE) expressions for CF-mMIMO under active pilot attacks and compared them with those of co-located massive MIMO systems. For the same system setup of~\cite{Timilsina:GC:2018}, the authors in~\cite{Hoang:TCOM:2018} discussed  power allocation problems either to maximize the achievable rate of the attacked legitimate user or to maximize the achievable SSE. Later, the authors in~\cite{Zhang:SYS.2020} investigated the effect of hardware impairments on the SSE of a CF-mMIMO network under pilot spoofing attacks. In~\cite{Alageli:TIFS:2020}, the problem of joint power and data transfer in a CF-mMIMO system with active eavesdropping was investigated.
Moreover, the secrecy performance of CF-mMIMO with non-orthogonal
pilot sequences for the uplink channel estimation was studied in~\cite{Salah:Globecom:2020}.
However, current studies tend to investigate the secrecy performance of CF-mMIMO systems with single-antenna APs, while  CF-mMIMO can better reap the channel hardening effect of cellular massive MIMO, when deploying multiple antennas at the APs~\cite{interdonato2019ubiquitous}. Therefore, the recent works of~\cite{Zhang:TVT:2010,Zhang:TCOM:2022}  studied the secrecy performance of multi-antenna CF-mMIMO networks under active eavesdropping. However, they only focused on the simple maximum-ratio
transmission (MRT) precoding scheme for downlink transmissions which is incapable of mitigating inter-user interference. Providing an analytical framework to characterize the secrecy performance of CF-mMIMO over active eavesdropping by applying more advanced distributed precoding techniques  is, therefore, of paramount importance and is one main goal of this paper. 

Given that not all the serving APs in CF-mMIMO will contribute equally to the SE per user, due to the path loss which sharply changes with the geographical distance, several AP selection schemes  for the downlink transmissions have been proposed in~\cite{Hien:Asilomar:2018,Ammar:GLOBSIP:2019,Buzzi:TWC:2020,Chen:JSAC:2021}. It has been shown that AP selection can effectively improve both the  SE and energy efficiency while preserving the system scalability. However, a common assumption in  the current papers on secure CF-mMIMO systems is that all APs transmit to all the users in the coverage area. In practice, there are some APs which are located far away from the legitimate users and will not contribute much to the SE. Therefore, how to efficiently select APs to serve the users under an  active spoofing attack is of practical importance, but is still an open problem.
%===================
\begin{figure*}
\centering 
\begin{minipage}{.4\textwidth} 
\centering 
\includegraphics[width=0.9\linewidth]{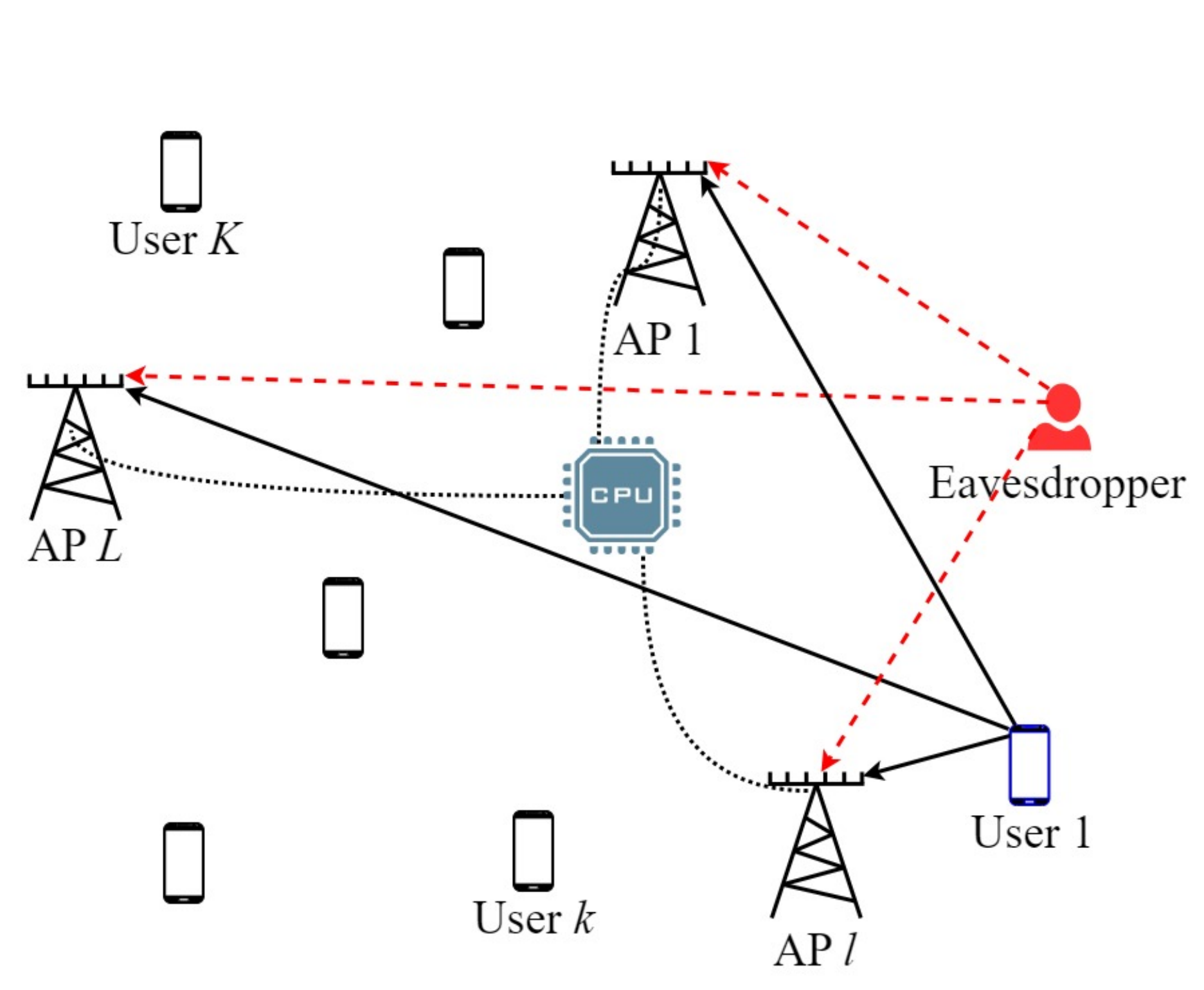}
\subcaption[a]{ Uplink training phase}
\label{Uplink Training Phase} 
\end{minipage}% 
\begin{minipage}{.4\textwidth} 
\centering \includegraphics[width=0.9\linewidth]{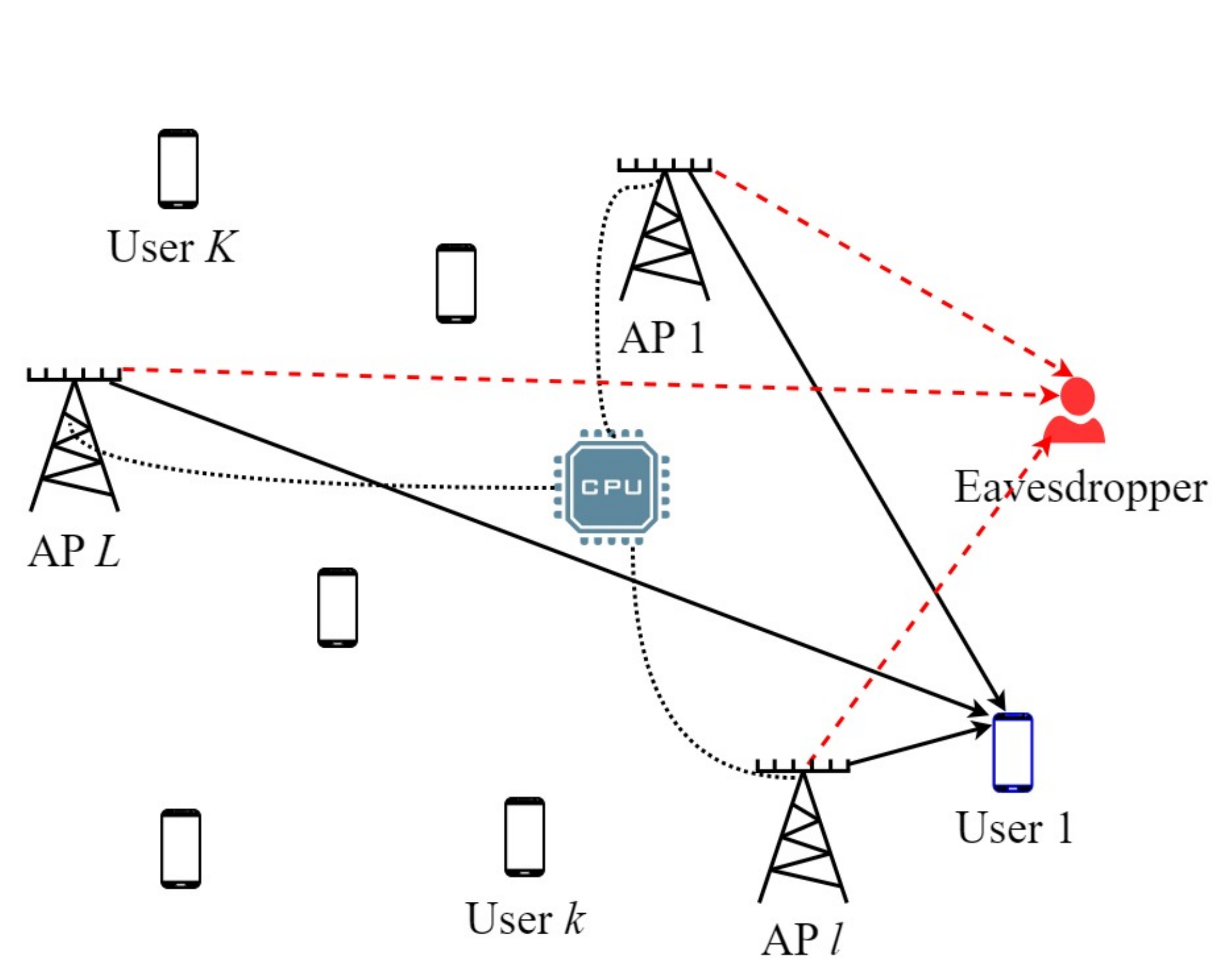}
\subcaption[b]{Downlink data transmission phase}
\label{Downlink Phase} 
\end{minipage} 
\caption{CF-mMIMO with $L$ multiple-antenna APs and $K$ legitimate users under an active  spoofing attack, where an eavesdropper contaminates the uplink channel estimation phase by sending an identical pilot sequence with the legitimate user $1$. 
}\label{fig1}
\hrulefill
\end{figure*}

We would like to highlight that this work is an extension of~\cite{YASSEEN:2023:VTC}.
Specifically,~\cite{YASSEEN:2023:VTC} pursued a performance analysis of  CF-mMIMO systems with multi-antenna APs experiencing an active eavesdropping attack, while the quality-of-service (QoS) SE requirements for users and maximum allowable SINR at the eavesdropper were ignored, and the impact of  AP selection was not investigated either. 
More importantly, we also provide an efficient scheme to detect the presence of eavesdroppers and determine which user is attacked. The main contributions of this paper are as follows:

%-------------------

\begin {itemize}
\item We develop a  framework for a CF-mMIMO system with multiple-antenna APs employing protective partial ZF-based (PPZF) precoding~\cite{Emil:TWC:2020}  under an active spoofing attack during the uplink training phase, while we also consider channel estimation errors.
\item We derive  closed-form expressions for the SSE which sheds useful insights into the system's performance. We then propose a greedy large-scale-based AP selection scheme to improve the SSE. 
\item We formulate an optimization problem for maximizing the received SINR at the user under attack, subject
to a maximum allowable SINR at the eavesdropper and maximum
transmit power at each AP while guaranteeing the specific QoS requirements on other users. We solve the optimization problem using the path-following algorithm.  We propose a novel and simple scheme to detect the presence of eavesdroppers in our system and determine which user is attacked. The method is based on the sample average power of the received pilot signals and can be implemented distributively at each AP, providing a detection probability of nearly one. We also propose a pilot re-transmission scheme to suppress the effect of the pilot spoofing attack.
We then extend our results to the multiple-antenna eavesdropper scenario. Finally, numerical results are presented to support our
findings and investigate the effect of different system parameters, including the number of antennas and APs and the position of the eavesdropper on the SSE performance, of CF-mMIMO systems. 
\end {itemize}

\textit{Notation:} We use bold upper case letters to denote matrices, and lower case letters to denote vectors. The superscript $(\cdot)^\text{H}$ stands for the conjugate-transpose (Hermitian);  $\mathbb{C}^{L\times N}$ denotes a $L\times N$ matrix; $\qI_M$  represents the $M\times M$ identity matrix; $\trace(\cdot)$ denotes the trace operation. A zero mean circular symmetric complex Gaussian distribution having variance $\sigma^2$ is denoted by $\mathcal{CN}(0,\sigma^2)$. Finally, $\Ex\{\cdot\}$ denotes the statistical expectation.	
%-------------------------------------
%%%%%%%%%%%%%%%%%%%%%%%%%%%%%%%%%%%%%%%%%

%%%%==============================================================================
\begin{small}
\begin{table}
\centering
 \caption{List of Notations}
%\vspace{-0.5em}
%\centering
\begin{tabular}{|c|c|}
\hline \text {Notation} & \text { Description } \\
\hline $L$ & \text {Number of APs} \\
\hline $M$ & \text {Number of AP antennas} \\
\hline $\betalk$ &\text {Large-scale fading coefficient for AP $l$ and UE $k$}\\
\hline $\glmk$ &\text {Small-scale fading coefficient for AP $l$ and UE $k$}\\
\hline $\gamma_{l, k}$ &\text {Variance of channel estimation for AP $l$ and UE $k$}\\
\hline $\Bphik$ & \text {Pilot sequence of the $k$-th user} \\
\hline $\tau_p$  & \text {Uplink training duration} \\
\hline $\BphiE$ & \text {Pilot sequence sent by the eavesdropper} \\
\hline $\hlmk$ & \text {Channel vector between AP $l$ and UE $k$} \\
\hline $\hlmE$ & \text {Channel vector between AP $l$ and the eavesdropper} \\
\hline $\mathbf {w}^ {\PZF} _{l, k}$ & \text {PZF precoding vector} \\
\hline $\mathbf {w}^ {\PMRT} _{l, k}$ & \text {PMRT precoding vector} \\
\hline $\rholk$ & \text {Power control coefficient for AP $l$ and UE $k$} \\
\hline $R_{sec}$ & \text {Secrecy spectral efficiency} \\
\hline
\end{tabular}
\label{t1}
\end{table}
\end{small}
%%%%==============================================================================
\section{System Model}~\label{sec:sysmodel}
%-------------------------------------
We consider a CF-mMIMO system comprising $L$ APs and $K$ single-antenna users. To deal with multiple-antenna users, one possible approach is to consider each user's antenna as an independent user~\cite{Marzetta:Cambridge:2016}. This system also contends with a single-antenna active eavesdropper, denoted as $E$, as illustrated in Fig.~\ref{fig1}. All APs cooperate to send data to all $K$ users, while the eavesdropper attempts to intercept the information intended for one of $K$ users by a pilot spoofing attack scheme.\footnote{Active full-duplex (FD) eavesdroppers are capable of both jamming as well as eavesdropping and/or  sending spoofing pilot sequences~\cite{Mohammadi:ICST:2023}. In this paper, we consider that an active half-duplex (HD) eavesdropper  focuses on sending the spoofing pilot sequences. The reason is that by attacking only the channel estimation phase, the eavesdropper remains fairly covert and power conservative as she only needs to operate during the estimation phase,  which  is typically  a small fraction of a  payload data transmission phase.}
More specifically, the eavesdropper sends an identical pilot sequence with a legitimate user of interest. By sending the same pilot with the legitimate user, the channel estimation at the APs is biased and contains the components of eavesdropping channel. In the downlink data transmission phase, since the precoding vectors at the APs associated with the targeted legitimate user are designed based on this biased channel estimation, the transmitted data signals may divert away from the legitimate user under attack and towards the direction of the eavesdropper. Therefore, the pilot spoofing attack leads to performance degradation of the legitimate transmission and, more severely, to information leakage toward the eavesdropper. 
We assume that the system knows the presence of an eavesdropper, and knows which user is being targeted. The method to achieve the above information is discussed in Section~\ref{sec:Eav_detection}. In addition,  we  consider an intelligent  eavesdropper with protocol knowledge and the ability to
synchronize and target the channel estimation phase.
In particular, given the recent advancements in software defined radios (SDRs), it is straightforward for smart eavesdroppers to target the pilots in a synchronous fashion for synchronous protocols~\cite{Miller:Mob:2012}. Without loss of generality, we assume that the eavesdropper targets user $1$. 

The sets of APs and users are denoted by $\mathcal{L}=\{1,\ldots, L\}$ and $\mathcal{K}=\{1,\ldots, K\}$, respectively. Each AP is equipped with $M$ antennas such that $LM > K$ and all APs are distributed over a certain area.
The channels of our considered system are modeled as follows:
\begin{itemize}
    \item The $M\times 1$ channel vector between the $l$-th AP and the $k$-th user is:
    \begin{align}~\label{eq:hlmk}
		\hlmk=\sqrt{\betalk} \glmk,
    \end{align}
where $\betalk$ is the large-scale fading coefficient, and $\glmk\sim \mathcal{CN}({\bf 0},\qI_M)$ is the small-scale fading vector.

    \item The $M\times 1$ channel vector between the $l$-th AP and the  eavesdropper is:
    \begin{align}~\label{eq:hlmk}
    \hlmE=\sqrt{\betalE} \glmE,
    \end{align}
where $\betalE$ represents the large-scale fading coefficient and $\glmE\sim \mathcal{CN}({\bf 0},\qI_M)$ is small-scale fading vector.
\end{itemize}

We assume that the considered CF-mMIMO system operates under time-division duplex (TDD) operation, where each coherence block includes three main phases: uplink training phase for channel estimation, downlink payload data transmission, and uplink payload data transmission. In this work, we focus on the downlink transmission, and hence, the uplink payload transmission phase is ignored.

\subsection{Uplink Training}

Uplink training phase is required for channel acquisition at the APs. These acquired channel estimates play a crucial role in both downlink precoding and uplink combining designs.
In this phase,  all users send pilot signals to the APs. Accordingly, each AP can estimate the corresponding channels to all users using the obtained pilot signals. Let us assume that the  $k$-th user sends a pilot sequence $\Bphik\in\mathbb{C}^{\tau_p \times 1}$. Here, $\tau_p$ represents the training duration. We consider orthonormal pilot assignment, i.e.,  $\Bphik^{\mathrm{H}}\Bphikp=0$ for $k\neq k'$ and $\|\Bphik\|^{2}=1$. This requires $\tau_p\geq K$.  

As previously mentioned, the eavesdropper aims at overhearing the confidential information intended for user $1$. This is accomplished by the eavesdropper transmitting a pilot sequence denoted as $\BphiE$, deliberately matching the pilot sequence sent by user 1, i.e., $\BphiE=\Bphione$. For most practical applications, the pilot sequences are publicly known and typically specified in the standard. Moreover, in some scenarios, some untrusted (malicious) nodes in the network may be regarded as eavesdroppers. In these scenarios, they are operating with known network protocols \cite{Zhou:TWC:2012,Hoang:TCOM:2018}. Accordingly, it is reasonable to assume that the pilot sequences associated with legitimate users are known to malicious users. Therefore, the pilot sequences $\Bphione, \ldots, \boldsymbol{\phi}_K$ associated with user $1$ to $K$  can be easily obtained by the eavesdropper. Then, the received pilot matrix at the $l$-th AP is given by
%====================
%\vspace{-0.8em}
	\begin{align}~\label{eq:yplm}
		\mathbf{Y}_{p, l}=\sqrt{\tau_p \rho_{u}} \sum_{k=1}^{K} \hlmk \Bphik^{\text{H}}+\sqrt{\tau_p \rho_{E}} \qh_{l, E} \Bphione^{\text {H}}+ \mathbf{N}_{l}^{},
	\end{align}
%------------
where $\rho_{u}$ and $\rho_{E}$ are the transmit signal-to-noise ratios (SNRs) at each user and the eavesdropper, respectively. More precisely, $\rho_{u} \triangleq P_{u} / N_{0}$ and $\rho_{E} \triangleq P_{E} / N_{0}$, where $P_{u}$ and $P_{E}$ are the transmit powers, while $N_{0}$ is the average noise power. In addition,  the noise matrix $\mathbf{N}_{l} \in \mathbb{C}^{M\times \tau_p}$ includes independent and identically distributed (i.i.d.) $\mathcal{CN}(0,1)$ elements.
From the received pilot signal  \eqref{eq:yplm}, AP $l$ uses the minimum mean squared (MMSE) estimation technique to estimate the channels to all users. The MMSE channel estimate of $\hlmk$ is given by
%===========================
\begin{align}~\label{eq:hhatlmk}  
		\hhatlmk = 
  \begin{cases} 
  \frac {\sqrt {\tau_{p}\rho _{u}}\betalk}{\tau_{p}\rho _{u}\betalk +1} \mathbf y_{l,k}, & k\neq 1, \\ 
  \frac {\sqrt {\tau_{p}\rho _{u}}\betalone}{\tau_{p}\rho _{u}\betalone + \tau_{p}\rho _{E}\betalE + 1} \mathbf y_{l,1}, & k= 1, 
  \end{cases}
\end{align}
%================================
where $\mathbf y_{l, k}=\mathbf Y_{p, l}\Bphik$. From \eqref{eq:hhatlmk}, we can see that $\hhatlmk$ includes $M$ i.i.d. $\mathcal{CN}(0,\gamalmk)$ elements, where
\begin{align}~\label{eq:gamalmk}  
		\gamalmk= 
  \begin{cases}
  \frac{\tau_{p} \rho_{u} \betalk^{2}}{\tau_{p} \rho_{u} \betalk+1}, & k \neq 1, 
  \\ 
  \frac{\tau_{p} \rho_{u} \betalone^{2}}{\tau_{p} \rho_{u} \betalone+\tau_{p} \rho_{E} \betalE+1}, 
  & k=1.
		\end{cases} 
\end{align}
%=============
Furthermore, let $\tilde { \mathbf {h}}_{l,k}$ be the channel estimation error, i.e., $\tilde { \mathbf {h}}_{l,k}=\hlmk-\hhatlmk$. Then, $\tilde { \mathbf {h}}_{l,k}$ is independent of $\hhatlmk$, and includes i.i.d. $\mathcal{CN}(0,\betalk-\gamalmk)$ elements.
%%%%%%%%%%%%%%%%%%%%%%%%%%%%%%%%%%%%%%%%%%%%%
\subsection{Downlink Data Transmission}
%%%%%%%%%%%%%%%%%%%%%%%%%%%%%%%%%%%%%%%%%%%
The channels, which are estimated during the uplink training phase, will serve as the basis for precoding the data symbols intended for all $K$ users. The precoded signal transmitted by  the $l$-th AP can be expressed as
%===========================
\begin{align}~\label{eq:x_l} 
		\mathbf{x}_{l}=  \sum_{k=1}^{K} \sqrt {\rho _{l,k}} \qw_{l, k} s_{k},
\end{align}
%===========================
where $\qw_{l, k} \in \mathbb{C}^{M \times 1}$, where $\Ex\{\|\qw_{l, {k}}\|^{2}\}=1$, is the precoding vector used by AP $l$ towards user $k$, and $\rho _{l,k}$ is the power control coefficient. Moreover, $s_{k}$ denotes the data symbol intended for the $k$-th user, where $\Ex\{|s_{k}|^{2}\}=1$. Thus, the received signals at the $k$-th user and the eavesdropper are respectively given by
%===========================
%\vspace{-0.5em}
\begin{align}~\label{eq:z_k}   
{z}_{k}= \sum_{l=1}^{L}\mathbf h^{\mathrm{H}}_{l,k}\mathbf{x}_{l}+{n}_{k},
\end{align}
%===========================
and
%===========================
	\begin{align}~\label{eq:z_E} 
		{z}_{E}= \sum_{l=1}^{L}\mathbf h^{\mathrm{H}}_{l,E}\mathbf{x}_{l}+{n}_{E},
	\end{align}
%===========================
% %===========================
% %\vspace{-0.5em}
% 	\begin{align}~\label{eq:z_k}   
% 		{z}_{k}= \sum_{l=1}^{L}\mathbf h^{\mathrm{H}}_{l,k}\bigg(\sum_{t=1}^{K} \sqrt {\rho _{l,k}} \mathbf {w}_{l,t} s_{t}\bigg)+{n}_{k},
% 	\end{align}
% %===========================
% and,
% %===========================
% 	\begin{align}~\label{eq:z_E} 
% 		{z}_{E}= \sum_{l=1}^{L}\mathbf h^{\mathrm{H}}_{l,E}\bigg(\sum_{t=1}^{K} \sqrt {\rho _{l,k}} \mathbf {w}_{l,t} s_{t}\bigg)+{n}_{E},
% 	\end{align}
% %===========================
where $n_{k} \sim \mathcal{CN}(0,1)$ and $n_{E} \sim \mathcal{C N}(0,1)$ are the corresponding additive noise terms.
%%%%%%%%%%%%%%%%%%%%%%%%%%%%%%%%%%%%%%%%%%%%%%%%%%%	
\subsection{Precoding Design}
%%%%%%%%%%%%%%%%%%%%%%%%%%%%%%%%%%%%%%%%%%%%%%%%%%%	
Classical ZF precoding for CF-mMIMO systems suffers from modest array gain due to the fact that almost all the available degrees of freedom (DoF) are used to mitigate the interference. On the other hand, while MRT effectively sustains the system scalability, it is unable to cancel inter-user interference. For these reasons, we adopt the PPZF precoding scheme~\cite{Emil:TWC:2020} to design the precoding vector $\qw_{l, k}$ in \eqref{eq:x_l}. The main idea of PPZF  is that each AP only mitigates the interference it causes to the strongest users, namely the users with the largest channel gain, while the interference towards the weak users, namely the users with the smallest channel gain, is tolerated. Consequently, PPZF can provide an acceptable trade-off between interference cancellation and boosting the desired signal.

More specifically, in PPZF scheme, each AP $l$ virtually divides users into two groups: \emph{1)} strong user set, and \emph{2)} weak user set based on their channel gains, $\Betalk$, $\forall k \in \mathcal{K}$. The user grouping can adhere to various criteria. For example, one criterion could be the mean square of the channel gain: user $k$ is assigned to the strong group for AP $l$  if $\Betalk$ exceeds a predetermined threshold; otherwise, user $k$ belongs to a weak group. Now, let us denote, for AP $l$,  the set of indices of strong users by $\Sl \subset \mathcal{K}$, and the set of indices of weak users by $\mathcal{W}_{l} \subset \mathcal{K}$, respectively, where $\Sl \cap \mathcal{W}_{l}=\varnothing$, and $\vert\Sl\vert+\vert\mathcal{W}_{l}\vert=K$. Then, AP $l$ transmits to the users in  $\Sl$ by using partial ZF, and to the users in $\mathcal{W}_{l}$  by using protective MRT (PMRT) to avoid interference to the users in $\Sl$. Here, we note that to implement PPZF, the number of transmit antennas at each AP, must meet the requirement $M\geq |\Sl|$.

{Let $\hat\qH_l = [\hat\qh_{l,1}, \ldots, \hat\qh_{l,K}] \in\mathbb{C}^{ M \times K}$ be the collective channel estimation matrix from AP $l$ to all users,  $\mathbf{E}_{\Sl}=	\big[\qe_i: i \in \Sl\big] \in \mathbb{C}^{ K \times |\Sl|}$, where $\mathbf{e}_i$ is the $i$-th column of $\mathbf{I}_{K}$. In addition, let user $k$ correspond to the $j$-th element of set $\Sl$, $j \in \{1, \ldots , |\Sl|\}$. Then, we define $\Bvarepsilon_{k}$ as the $j$-th column of $\mathbf{I}_{|\Sl|}$. With PPZF, the precoding vector in \eqref{eq:x_l} can be expressed as }
\begin{align}~\label{eq:wlk_PPZF1}  
		\qw_{l, k}= 
  \begin{cases}
  \mathbf {w}^ {\PZF} _{l, k}, & \text{if} ~ k \in \mathcal {S} _{l}, 
  \\ 
  \mathbf {w}^ {\PMRT} _{l,k}, 
  & \text{if} ~ k \in \mathcal {W} _{l},
		\end{cases} 
\end{align}
where $\mathbf {w}^ {\PZF} _{l, k}$ is the PZF precoding vector and $\mathbf{w}^ {\PMRT} _{l,k}$ is the PMRT precoding vector. More precisely,
\begin{itemize}
    \item PZF precoding vector:
    %=============================
\begin{align}~\label{eq:wPZF}
\qw_{l, {k}}^{\PZF}=\frac{\hat{\qH}_{l} \mathbf{E}_{\mathcal{S}_{l}}\left(\mathbf{E}_{\mathcal{S}_{l}}^{\mathrm{H}} \hat{\qH}_{l}^{\mathrm{H}} \hat{\qH}_{l} \mathbf{E}_{\mathcal{S}_{l}}\right)^{-1} \Bvarepsilon_k}
{\sqrt{\Ex\Bigg\{\bigg\|\hat{\qH}_{l} \mathbf{E}_{\mathcal{S}_{l}}\left(\mathbf{E}_{\mathcal{S}_{l}}^{\mathrm{H}} \hat{\qH}_{l}^{\mathrm{H}} \hat{\qH}_{l} \mathbf{E}_{\mathcal{S}_{l}}\right)^{-1} \Bvarepsilon_{k}\bigg\|^{2}\Bigg\}}},
\end{align}
%==============================
where the denominator of \eqref{eq:wPZF} is the normalization term, which is given in closed-form as
%==============================
\begin{align} ~\label{eq:normterm}
\Ex\left\{\left\|\hat{\qH}_{l} \mathbf{E}_{\mathcal{S}_{l}}\left(\mathbf{E}_{\mathcal{S}_{l}}^{\mathrm{H}} \hat{\qH}_{l}^{\mathrm{H}} \hat{\qH}_{l} \mathbf{E}_{\mathcal{S}_{l}}\right)^{\!\!-1
}\!\! \Bvarepsilon_k\right\|^{2}\right\}\!=\!\frac{1}{\left(M\!-\!{\tausl}\right) \gamma_{l, k}}.
\end{align}
%===============================
Therefore, for any pair of users $k, t \in \Sl$ we have
%===========================
\begin{align} ~\label{eq:alphaPZF} 
	\alpha^{\textsf {PZF}} _{l,k,t}\triangleq&\hat { \qh}_{l,k}^{\text {H}} \mathbf {w} ^{ \textsf {PZF} }_{l,t} =\begin{cases} 0, & t \neq{k}, \\ \sqrt {(M-\tausl)\gamma _{l,k}}, & t ={k}.
 \end{cases}\end{align}
    \item PMRT precoding vector, from AP $l$ to user $j$, $j \in \mathcal{W}_{l}$:
     \begin{equation}\label{eq:wPMRT}
\mathbf {w}^ {\PMRT} _{l,j} = \frac { \mathbf {B}_{l}\hat 
\qH_{l} \mathbf {e}_{j} } {\sqrt { \Ex\{ \|\mathbf {B}_{l}\hat { \mathbf {H}}_{l} \mathbf {e}_{j}\|^{2} \} }} = \frac { \mathbf {B}_{l}\hat { \mathbf {H}}_{l} \mathbf {e}_{j} } {\sqrt { (M- \tausl)\gamma_{l,j} }},
 \end{equation}
%===========================
where
%===========================
	\begin{align}~\label{eq:Bl}
		\mathbf{B}_{l}=\mathbf{I}_{M}-\hat{\qH}_{l} \mathbf{E}_{\mathcal{S}_{l}}\left(\mathbf{E}_{\mathcal{S}_{l}}^{\mathrm{H}} \hat{\qH}_{l}^{\mathrm{H}} \hat{\qH}_{l} \mathbf{E}_{\mathcal{S}_{l}}\right)^{-1} \mathbf{E}_{\mathcal{S}_{l}}^{\mathrm{H}} \hat{\qH}_{l}^{\mathrm{H}}, 
	\end{align}
%===========================
which is the null space of $\hat{\qH}_l \mathbf{E}_{\Sl}$. For any pair of users $k, t \in \mathcal{W}_{l}$ we have
%===========================
\begin{align}~\label{eq:Ex} 
		\Ex\left \{{{\hat { \qh}_{l,k} ^{\text {H}} \mathbf {w} ^ {\textrm {PMRT}} _{l,{t}}}}\right \} = \begin{cases} 0, & t \neq{k}, \\ \sqrt {(M-\tausl) \gamma _{l,k}}, & t ={k}. \end{cases}
\end{align}
%===========================
\end{itemize}
Then, the transmit signals at AP $l$  is
%============================
\begin{align}~\label{eq:xl}
\mathbf {x}_{l} = \sum _{k \in \mathcal {S} _{l}} \sqrt {\rho _{l,k}} \mathbf {w}^ {\PZF} _{l, k} s_{k} + \sum _{j \in \mathcal {W} _{l}} \sqrt {\rho _{l,j}} \mathbf {w}^ {\PMRT} _{l,j} s_{j}.
\end{align}
%=============================

\begin{figure*}
\begin{align}
~\label{eq:SINRk}
\SINRk=
&\frac{\bigg|\sum\limits _{l \in \mathcal {Z} _{k}}\!\sqrt{\rholk} \Ex\!\left\{\qh_{l, k}^{\mathrm{H}}\qw_{l, {k}}^{{\PZF}}\right\}+\sum\limits_{p\in \mathcal{M}_{k}}\!\! \sqrt{\rhopk} \Ex\!\left\{\qh_{p, k}^{\mathrm{H}} \qw_{p, k}^{{\PMRT}}\right\}\bigg|^{2}}
{\sum\limits_{t=1}^{K}\!\Ex\Bigg\{\bigg|\sum\limits _{l \in \mathcal {Z} _{t}}\!\!  \sqrt{\rholt} \qh_{l, k}^{\mathrm{H}} \qw_{l, {t}}^{\PZF}\!+\!\!\sum\limits _{p\in \mathcal {M} _{t}}\!\!\! \sqrt{\rhopt} \qh_{p, k}^{\mathrm{H}} \qw_{p, {t}}^{{\PMRT}}\bigg|^{2}\!\Bigg\}-\bigg|\sum \limits_{l \in \mathcal {Z}_{k}}\!\!\!\!\sqrt{\rholk} \Ex\!\left\{\!\qh_{l, k}^{\mathrm{H}} \qw_{l, {k}}^{\PZF\!}\right\}\!+\!\!\!\sum\limits _{p\in \mathcal {M} _{k}}
\!\!\!\sqrt{\rhopk} \Ex\left\{\!\qh_{p, k}^{\mathrm{H}} \qw_{p,k}^{{\PMRT}}\right\}\bigg|^{2}+1}.\tag{25}
\end{align}
	\hrulefill
\end{figure*}
%==================

%===========================
%%%%%%%%%%%%%%%%%%%%%%%%%%%%%%%%%%%%%%%%%%%%%%%%%%%
\section{Secrecy Performance Analysis}
%%%%%%%%%%%%%%%%%%%%%%%%%%%%%%%%%%%%%%%%%%%%%%%%%%%%%
In this section, we evaluate the secrecy performance provided by a CF-mMIMO system with PPZF precoding under an active eavesdropping attack. 
First, let us define $\mathcal{Z}_{k}$ and $\mathcal{M}_{k}$ as the set of indices of APs that transmit to the user $k$ by using PZF and PMRT, respectively, as 
%-----
\begin{align}\label{eq:Z_k}
\mathcal{Z}_{k} \triangleq\left\{l: k \in \mathcal{S}_{l}, l=1, \ldots, L\right\}, 
\end{align}
%----
and
%----
\begin{align}\label{eq:M_k}
\mathcal{M}_{k} \triangleq\left\{l: k \in \mathcal{W}_{l}, l=1, \ldots, L\right\},
\end{align}
%======
with $\mathcal{Z}_{k} \cap \mathcal{M}_{k}=\varnothing$, and $\left|\mathcal{Z}_{k}\right|+\left|\mathcal{M}_{k}\right|=L$. 

%==================================================
\subsection{Spectral Efficiency of the Legitimate Links}
Using~\eqref{eq:Z_k} and~\eqref{eq:M_k}, the received signal at the $k$-th user in~\eqref{eq:z_k} can be rewritten as
%==========================
\begin{align}\label{eq:z}
&z_{k}=\Bigg(\sum_{l \in \mathcal{Z}_{k}} \sqrt{\rholk} \qh_{l, k}^{\mathrm{H}} \qw_{l, {k}}^{\PZF}\notag+\sum_{p \in \mathcal{M}_{k}} \sqrt{\rhopk} \qh_{p, k}^{\mathrm{H}} \qw_{p, {k}}^{\PMRT}\Bigg) s_{k} \\
&+\sum_{\substack{t=1 \\ t \neq k}}^{K}
\Bigg(\!\sum_{l \in \mathcal{Z}_{t}} \!\sqrt{\rholt} \qh_{l, k}^{\mathrm{H}} \qw_{l, {t}}^{\PZF}\notag\!+\!\!\sum_{p \in \mathcal{M}_{t}}\! \sqrt{\rhopt} \qh_{p, k}^{\mathrm{H}} \qw_{p, {t}}^{\PMRT}\Bigg) s_{t} \!+\!n_{k},\nonumber\\
% z_{k} &= \sum _{l = 1}^{L} \sqrt { \rho _{l,k} } \mathbf {h}_{l,k} ^{\text {H}}\mathbf {w}_{l, k} s_{k}  +  \sum_{\substack{t=1 \\ t \neq k}}^{K}\sum \limits _{l=1}^{L}\sqrt { \rho _{l,t} } \mathbf {h}_{l,k} ^{\text {H}}\mathbf {w}_{l,t} s_{t} + n_{k},\nonumber\\
& =\text {CP}_{k} \times s_{k} + \text {PU}_{k} \times s_{k} + \sum_{\substack{t=1 \\ t \neq k}}^{K}\text {UI}_{k,t} \times s_{t} + n_{k},
\end{align}
%=======================
where $\mathrm{CP}_{k}$, $\mathrm{PU}_{k}$, and $\mathrm{UI}_{k,t}$ show the coherent precoding gain, precoding gain uncertainty, and multi-user interference, respectively, defined as
 %===================
\begin{align}~\label{eq:z_k_3}
\text {CP}_{k}&\!\!= \!\!\!\sum_{l \in \mathcal{Z}_{k}}\!\! \sqrt{\rholk} \mathop {\Ex}\nolimits \left \{{\qh_{l, k}^{\mathrm{H}} \qw_{l, {k}}^{\PZF}}\right \}\!\!+\!\!\!\!\sum_{p \in \mathcal{M}_{k}} \!\!\sqrt{\rhopk} \mathop {\Ex}\nolimits \left \{{\qh_{p, k}^{\mathrm{H}} \qw_{p, {k}}^{\PMRT}}\right \}, \\ \text {PU}_{k}&=\sum_{l \in \mathcal{Z}_{k}}\left ({\sqrt {\rho _{l,k}} \mathbf {h}_{l,k} ^{\text {H}}\qw_{l, {k}}^{\PZF} - \sqrt {\rho _{l,k}} \mathop {\Ex}\nolimits \left \{{\mathbf {h}_{l,k} ^{\text {H}}\qw_{l, {k}}^{\PZF} }\right \}}\right)\notag\\ &+\sum_{p \in \mathcal{M}_{k}} \left ({\sqrt {\rho _{l,k}} \mathbf {h}_{l,k} ^{\text {H}}\qw_{p, {k}}^{\PMRT} - \sqrt {\rho _{l,k}} \mathop {\Ex}\nolimits \left \{{\mathbf {h}_{l,k} ^{\text {H}}\qw_{p, {k}}^{\PMRT}}\right \}}\right), \\
\text {UI}_{k,t}&=\sum_{l \in \mathcal{Z}_{k}}\!\! \sqrt{\rholt} {\qh_{l, k}^{\mathrm{H}} \qw_{l, {t}}^{\PZF}}\!\!+\!\!\!\!\sum_{p \in \mathcal{M}_{k}} \!\!\sqrt{\rhopt} {\qh_{p, k}^{\mathrm{H}} \qw_{p, {t}}^{\PMRT}}~\label{eq:{UI}_{k,t}}
. 
\end{align}
%================
Each user $k$ effectively sees a deterministic channel ($\mathrm{CP}_{k}$) with some unknown noise and, hence, to detect the intended symbol from the received
signal, it relies only on the
statistical channel state information (CSI). More specifically, since $s_k$ and $s_t$ are uncorrelated for any $t \neq k$, the first term in~\eqref{eq:z} is uncorrelated with the third term. Additionally, since $s_k$ is independent of $\text {PU}_{k}$, the first and second terms are also uncorrelated. The fourth term. i.e, noise, is independent of the first term in~\eqref{eq:z}. Accordingly, the sum of the second, third, and fourth terms in~\eqref{eq:z} can be considered as an uncorrelated effective noise. Following the discussion in~\cite [Sec. 2.3.2]{Marzetta:Cambridge:2016},  an achievable downlink SE for user $k$ can be written as
%=======
\begin{align} \label{eq:SE_k1} 
R_k= \log_2\left(1+\SINRk \right),
\end{align}
where
%================
\begin{align}
\SINRk=\frac{\left|\mathrm{CP}_k\right|^2}{\Ex\left\{\left|\mathrm{PU}_k\right|^2\right\}+\sum\limits_{t \neq k}^K \Ex\left\{\left|\mathrm{UI}_{k, t}\right|^2\right\}+1},
\end{align}
%=============
which can be re-expressed as~\eqref{eq:SINRk} at the top of the next page.
%  To detect the intended symbol from the received
% signal, each user is assumed to rely on the
% statistical channel state information. Therefore, by using the use-and-then-forget bounding technique, ${\color{red}{[R_3C_3]}}$\comm{which utilizes the mean value of the channel as the channel estimate instead of the full channel knowledge due to the property of channel hardening ~\cite{Marzetta:Cambridge:2016}}, an achievable SE of the $k$-th user is

% where  $\SINRk$ is given by~\eqref{eq:SINRk}, shown at the top of the next page. 
%=================================
% \begin{figure*}
% \begin{align}
% ~\label{eq:SINRk}
% \SINRk^{{\PPZF}}=
% &\frac{\left|\sum\limits _{l \in \mathcal {Z} _{k}}\!\sqrt{\rholk} \Ex\!\left\{\qh_{l, k}^{\mathrm{H}}\qw_{l, {k}}^{{\PZF}}\right\}+\sum\limits_{p\in \mathcal{M}_{k}}\!\! \sqrt{\rhopk} \Ex\!\left\{\qh_{p, k}^{\mathrm{H}} \qw_{p, j}^{{\PMRT}}\right\}\right|^{2}}
% {\sum\limits_{t=1}^{K}\!\Ex\left\{\left|\sum\limits _{l \in \mathcal {Z} _{t}}\!\!  \sqrt{\rholt} \qh_{l, k}^{\mathrm{H}} \qw_{l, {t}}^{\PZF}\!+\!\!\sum\limits _{p\in \mathcal {M} _{t}}\!\!\! \sqrt{\rhopt} \qh_{p, k}^{\mathrm{H}} \qw_{p, {t}}^{{\PMRT}}\right|^{2}\!\right\}-\left|\sum \limits_{l \in \mathcal {Z}_{k}}\!\!\!\!\sqrt{\rholk} \Ex\!\left\{\!\qh_{l, k}^{\mathrm{H}} \qw_{l, {k}}^{\PZF\!}\right\}\!+\!\!\!\sum\limits _{p\in \mathcal {M} _{k}}
% \!\!\!\sqrt{\rhopk} \Ex\left\{\!\qh_{p, k}^{\mathrm{H}} \qw_{p,j}^{{\PMRT}}\right\}\right|^{2}+1}
% \end{align}
% 	\hrulefill
% 	\vspace{-4mm}
% \end{figure*}
% %==================
We next provide a  closed-form SE expression for user $k$.
%=====================
\begin{proposition}\label{prop:SINRk}
The closed-form expression for the SE of user $k$ with PPZF precoding under an active eavesdropping attack  can be expressed as \eqref{eq:SE_k1}, where
%===========================
\setcounter{equation}{25}
\begin{align}~\label{eq:SINR_k2} 
\SINRk= \frac{\left(\!\!\!\!\quad\sum\limits_{l=1}^{L} \sqrt{\left(M-\tausl\right) \rholk \gamma_{l, k}}\right)^{2}}{\sum\limits_{t=1}^{K} \sum\limits_{l=1}^{L} \rholt\left(\beta_{l, k}-\delta _{l,k}\gamma_{l, k}\right)+1},
\end{align}
%===========================
where $\delta_{l, k} \triangleq 1$ if  $k\in \mathcal{S}_{l}$ and $\delta_{l, k} \triangleq 0$ if  $k\in \mathcal{W}_{l}$.
%===========================
 %===========================
\end{proposition}
\begin{proof}
See Appendix~\ref{app:SINRK}.
\end{proof}
%==========================

%%%%%%%%%%%%%%%%%%%%%%%%%%%%%%%%%%%%%%%%%%%%%%%%%%%%%

\subsection{Spectral Efficiency of  the Eavesdropper}
%%%%%%%%%%%%%%%%%%%%%%%%%%%%%%%%%%%%%%%%%%%%%%%%%%%%%
Hereafter, we assume that the eavesdropper has perfect  CSI knowledge which results in the 
worst case SSE. In particular, the received signal at eavesdropper~\eqref{eq:z_E} can be represented in the form of %===========================
\begin{align} ~\label{z_E2} 
z_{E} &= \sum\limits_{l=1}^{L} \!\!
\sqrt{\rho_{l, 1}}\mathbf h_{l,E}^{\text{H}}  \qw_{l,1} s_{1} 
\!+ \!\! {\sum\limits_{\substack{ t \neq 1}}^{K} } \! 
\sum\limits_{l=1}^{L} \!\!
\sqrt{\rholt} \mathbf h_{l,E}^{\mathrm{H}}  \qw_{l,t}s_{t}\! +\! n_{E},\notag\\
&=\BUE \times s_1+\underbrace{\sum_{t \neq 1}^K \mathrm {UI}_{ E, t} \times s_{t}+n_{E}}_{\text {treated as noise }},
\end{align}
%============
where $\BUE$ represents the strength of the desired signal $s_1$, while $\UIEt$ denotes the interference caused by the remaining users $(t\neq k)$, and they can be expressed as
%=============
\begin{align}~\label{eq:SINRE1}
  \BUE&\triangleq 
  \sum\limits_{l \in \mathcal{Z}_{1}} \sqrt{\rho_{l, 1}} \qh_{l, E}^{\mathrm{H}} \qw_{l, 1}^{\PZF}+\sum\limits_{p \in \mathcal{M}_{1}} \sqrt{\rho_{p, 1}} \qh_{p, E}^{\mathrm{H}} \qw_{p, 1}^{\PMRT},
  \\
  \mathrm {UI}_{ E, t}&\triangleq
  \sum\limits_{l \in \mathcal{Z}_{t}} \sqrt{\rholt} \qh_{l, E}^{\mathrm{H}} \qw_{l, t}^{\PZF}+\sum\limits_{p \in \mathcal{M}_{t}} \sqrt{\rho_{p,t}} \qh_{p, E}^{\mathrm{H}} \qw_{p,t}^{\PMRT}. 
\end{align} 
%=============
% We treat the terms  $\BUE$,  $\mathrm {UI}_{ E, t}$, and $n_{E}$ in~\eqref{z_E2}  as noise.
Let us denote the mutual information between $s_1$ and $z_{E}$ by $I_{E}\left(s_1 ; z_{E}\right)$. Then, the upper-bound for $I_{E}\left(s_k ; z_{E}\right)$ is given by
%========================
\begin{align} \label{I_E}
 I_{\mathrm{E}}\left(s_1 ; z_{E}\right)&\stackrel{(a)}{\leq} I_{E}\left(s_1 ; z_{E} \mid\left\{h_{l, k}\right\}_{l, k},\left\{\hat{h}_{l, k}\right\}_{l, k},\left\{h_{l,E}\right\}_l\right) \notag \\
 &=\mathbb{E}\left\{\log _2\left(1+\frac{\left|\mathrm{BU}_{E, 1}\right|^2}{\sum\limits_{t \neq 1}^K\left|\mathrm{UI}_{E, t}\right|^2+1}\right)\right\},\notag\\
 &\stackrel{(b)}{\approx} \log _2\left(1+\frac{\mathbb{E}\left\{\left|\mathrm{BU}_{E, 1}\right|^2\right\}}{\sum\limits_{t\neq 1}^K \mathbb{E}\left\{\left|\mathrm{UI}_{E, t}\right|^2\right\}+1}\right),
\end{align}
%==========
where the inequality (a) comes from the fact that the eavesdropper has  {perfect CSI knowledge}, while the approximation in (b)  follows from \cite[Lemma 1]{Zhang:JSTSP:2014}.

% %=====================
% \begin{figure*}
%   \begin{align}~\label{eq:SINRE}
% \SINRE= \frac { \Ex \left \{{{|\BUE|^{2}}}\right \} }{{\sum\limits_{\substack{ t \neq 1}}^{K} }\Ex \left \{{{|\mathrm {UI}_{ \textrm {E}, t}|^{2}}}\right \} + 1 }
% =\frac{\Ex\bigg\{\bigg|\sum\limits_{l \in \mathcal{Z}_{1}} \sqrt{\rho_{l, 1}} \qh_{l, E}^{\mathrm{H}} \qw_{l, 1}^{\PZF}+\sum\limits_{p \in \mathcal{M}_{1}} \sqrt{\rho_{p, 1}} \qh_{p, E}^{\mathrm{H}} \qw_{p, 1}^{\PMRT}\bigg|^{2}\bigg\}}{{\sum\limits_{\substack{t \neq 1}}^{K} }\Ex\!\bigg\{\bigg|\sum\limits_{l \in \mathcal{Z}_{t}} \sqrt{\rholt} \qh_{l, E}^{\mathrm{H}} \qw_{l, t}^{\PZF}+\sum\limits_{p \in \mathcal{M}_{t}} \sqrt{\rho_{p,t}} \qh_{p, E}^{\mathrm{H}} \qw_{p,t}^{\PMRT}\bigg|^{2}\bigg\}+1}. 
% \end{align} 
% 	\hrulefill
% \end{figure*}
%===========================
%===========================
%============================
\begin{proposition}\label{prop:SINRE}
The SE of the eavesdropper, which overhears the confidential messages destined for user $1$ with PPZF precoding,  can be approximated as $R_E\approx \log_2\left(1+\SINRE \right)$, where
%===========================
\begin{align}~\label{eq:SINR_E2} 
\SINRE\!=\!\frac {\left(\sum\limits_{l=1}^{L} \!\!\sqrt {{\rho _{l,1}}(M\!\!-\!\!\tausl)\gamma_{l,E}}\right)^{2}\!\!\!\!+\!\!\sum\limits_{l=1}^{L}\rho_{l, 1}\beta_{l,E}\!\!-\!\!\!\sum\limits_{l \in \Zone}\!\! \rho_{l, 1}\gamma_{l,E}}{{\sum\limits_{\substack{ t \neq 1}}^{K}} \sum\limits_{l=1}^{L}\rholt(\beta_{l,E}\!-\!\delta_{l,1}\gamma_{l,E})\!+\!1}. 
\end{align}
%===========================
\end{proposition}
\begin{proof}
See Appendix~\ref{app:SINRE}.
\end{proof}
%====================

%%%%%%%%%%%%%%%%%%%%%%%%%%%%%%%%%%%%%%%%%%%%

\subsection{Secrecy Spectral Efficiency}
%%%%%%%%%%%%%%%%%%%%%%%%%%%%%%%%%%%%%%%%%%%%
Using the derived  SE expressions, we can now calculate  the SSE associated with user $1$ as 
%---------------------------
\begin{align}~\label{eq:R_sec} 
R_{sec}= [R_1-R_E]^+  \approx  \left[\log_{2}\left(\frac{1+\SINRone}{1+ \SINRE} \right)\right]^+,
\end{align}
where  $[x]^+ = \max \{0, x\}$. 
%------------------------------

%=========================================================
%========================================================

\section{Access Point Selection and Power Optimization}\label{sec:APsel_PA}
\subsection{Access Point Selection}
In this subsection, we propose an AP selection
scheme which can increase the SSE. The proposed scheme is based on the following observations:
Firstly, for  user $1$, there are some APs which are located
far away and will not add significantly to
its overall SE. Secondly, there are some APs which are in close vicinity of eavesdropper. Serving user $1$ by these APs may result in high values for the overheard SINR.  

Therefore, in order to increase the SSE, user $1$ should not be served by all APs. To this end, we propose a greedy large-scale-based  AP selection scheme for choosing a group of APs for serving user $1$ in Algorithm~\ref{alg:GR}. We consider the ratio $\zeta(l)=\frac{\beta_{l, 1}}{\betalE}$ as the criterion, order the APs based on the $\zeta(l)$ in a descending order and then create the set $\mathcal{A}=\{l^{(1)}, \ldots, l^{(L)} \}$.
Let $\mathcal{T}$ be the set
of assigned APs to user $1$ and $R_{sec}(\mathcal{T})$ denote the dependence of the SSE on the different choices of $\mathcal{T}$. The key idea of Algorithm~\ref{alg:GR} is to iteratively select a subset of APs out of the ordered AP set $\mathcal{A}$  on the condition that a new AP assignment at each iterative step improves the SSE. A key characteristic of the proposed AP selection is that it only
requires the large scale fading (path loss) between the APs and eavesdropper. Additionally, since it is performed only on a large-scale time scale, it avoids the need of frequently performing the AP selection.
{\begin{remark}
   In Algorithm~\ref{alg:GR}, the dominant term in calculating  $R_{sec}$ given in~\eqref{eq:R_sec} is the double summation in the denominator of \eqref{eq:SINR_k2} whose complexity scales as $\mathcal{O}\left(L × K\right)$. The operation of computing $\zeta\left(l\right)$ has complexity $\mathcal{O}\left(L\right)$, while sorting the order index of $L$ APs is performed with complexity $\mathcal{O}\left(L \log L\right)$.  Therefore, the computational complexity of the algorithm is in the order of $\mathcal{O}\left(L × K\right)$.
\end{remark}
%----------------------------------
% \begin{algorithm}[H]
% \caption{The Greedy AP Selection}
% \begin{algorithmic}[1]
% \label{alg:GR}
% \STATE \text{Divide $\beta_{l, 1}$ by $\betalE$.}
% \STATE   {Utilizing the aforementioned ratio, arrange the access points (APs) in descending order.}
% \STATE {Calculate the SSE of the initial AP element in (2), then add it to a new AP set namely APset.}
% \STATE {Calculate the SSE by adding the next access point in (2), check weather the SSE increases or decreases. If it increases, add this AP to APset, otherwise neglect it, and repeat the process by adding a new AP from (2).}
% \STATE {Repeat (4) until all elements in (2) have been considered.}
% \STATE {Finally, the system will employ the APs of APset for serving the attacked user and eavesdropper.}
% \end{algorithmic}
% \end{algorithm}
%*************************************************
%************************************************
\begin{algorithm}[!t]
\caption{Greedy AP Selection}
\begin{algorithmic}[1]
\label{alg:GR} 
\STATE
\textbf{Initialize}:  Set  $\mathcal{T}=\emptyset$ and $R_{sec}^{0,\star}=0$.
\STATE Calculate the ratio $\zeta(l)=\frac{\beta_{l, 1}}{\betalE}$ for all APs. Then, order the APs based on the ratio $\zeta(l)$ in a descending order and create the set $\mathcal{A}=\{l^{(1)}, \ldots, l^{(L)} \}$. 
\FOR {$i=1$ to $L$}
\STATE  Calculate  $R_{sec}^{i,\star}=R_{sec}(\mathcal{T}\cup l^{(i)})$ 
\IF{$R_{sec}^{i,\star} > R_{sec}^{i-1,\star}$}
% \STATE{Update  $\mathcal{A}_{\ul}=\{\mathcal{A}_{\ul}\bigcup m^{\star}\}$ }
% \ELSE
\STATE Update $\mathcal{T}=\mathcal{T} \cup\{l^{(1)}\}$
\ENDIF
\ENDFOR
\end{algorithmic}
\end{algorithm}

\subsection{Power Optimization}
Here, we aim at optimally selecting the power coefficients $\rholk$  to maximize the SE (and hence, the SINR) at user $1$,  subject to a maximum
allowable SINR at the eavesdropper and maximum transmit power at each AP while guaranteeing
the specific QoS requirements on each user $k$, $k \in \mathcal{K}\backslash \{1\}$.
Note that we  consider average normalized power constraint at each AP, i.e., $\mathbb{E}\left\{\Vert \qx_l\Vert^2\right\}\leq \rho_{\max }$ with $\rho_{\max }=P_{\max } / N_0$, which, using~\eqref{eq:x_l}, can  be further
expressed as the following per-AP power constraint
%---------------
\begin{align} ~\label{eq:Ex_l} 
\sum_{k=1}^K \rholk \leq\rho_{\max }.
\end{align}
%--------------
More precisely, the power optimization problem is formulated as  
%------------------
\begin{subequations}\label{eq:Opt1}
\	\begin{align}
		&\underset{\{\rholk\}}{\mathrm{max}}\,\, \hspace{1em}
		 \,\,\SINRone,
		\\
  &~\text {s.t.}   \hspace{1em}\sum_{k=1}^K \rholk \leq\rho_{\max }, ~~ \quad \forall k, l, \\ 
    &   \hspace{3em} \SINRk  \geq\theta_{k }, ~~ \quad \forall k, k \neq 1,  \\
        &   \hspace{3em} \SINRE  \leq \theta_{E}, ~~ \quad 
  	\end{align}
 \end{subequations}
 %-------------------
 where $\theta_k$ and $\theta_{E}$ are the minimum SINRs required by user $k$  and the maximum allowable overheard SINR at eavesdropper, respectively.
 %--
In order to facilitate further analysis, let us denote   the power allocation coefficient matrix by $\boldsymbol{\Psi}$ with elements $\boldsymbol{\Psi}(l, k)=\sqrt{\rholk}~~ \forall l, k$. The $k$-th column vector of $\boldsymbol{\Psi}$ is denoted as
%---------------
\begin{align}~\label{eq:uk} 
\mathbf{u}_k=\boldsymbol{\Psi}(:, k)=\left[\sqrt{\rho_{1, k}}, \sqrt{\rho_{2, k}}, \ldots, \sqrt{\rho_{L, k}}\right]^T.
\end{align}
%=======================
Furthermore, we  define the following matrices and vectors
%=======================
\begin{align}~\label{eq:ak} 
\mathbf{a}_k  =&\Big[\sqrt{(M-\mathcal{S}_1)\gamma_{1, k}}, \sqrt{(M-\mathcal{S}_2)\gamma_{2, k}}, \ldots, \nonumber\\
&\qquad\sqrt{(M-\mathcal{S}_L)\gamma_{L, k}}\Big]^T,
\end{align}
%===========================
\begin{align}~\label{eq:Akk} 
\mathbf{A}_{k, k} & = \operatorname{diag}\left(\sqrt{\beta_{1, k}-\delta_{1,k} \gamma_{1, k}}, \ldots, \sqrt{\beta_{L, k} -\delta_{L,k}\gamma_{L, k}}\right),
\end{align}
%===============================
\begin{align}~\label{eq:AE} 
\mathbf{b}_{\mathrm{E}}  =&\Big[\sqrt{(M-\mathcal{S}_1)\gamma_{1, E}}, \sqrt{(M-\mathcal{S}_2)\gamma_{2, E}}, \ldots, \nonumber\\
&\qquad\sqrt{(M-\mathcal{S}_L)\gamma_{L, E}}\Big]^T,
\end{align}
%================================
and
%------------
\begin{align}~\label{eq:BE} 
\mathbf{B}_{E} & =\operatorname{diag}\left(\sqrt{\beta_{1, E}-\delta_{1,1} \gamma_{1, E}}, \ldots, \sqrt{\beta_{L, E} -\delta_{L,1}\gamma_{L, E}}\right).
\end{align}
%==================================
Accordingly, using~\eqref{eq:uk}-\eqref{eq:BE}, the received SINRs at user $1$ and eavesdropper in \eqref{eq:SINR_k2} and \eqref{eq:SINR_E2} can be rewritten as 
%==================================
\begin{align}~\label{eq:SINRk3} 
\SINRk & =\frac{\left(\mathbf{a}_k^T \mathbf{u}_k\right)^2 } {\varphi_k(\boldsymbol{\Psi})},
\end{align}
%====================================
and
%=================
\begin{align} ~\label{eq:SINRE3} 
\SINRE & =\frac{(\mathbf{b}_{\mathrm{E}} \mathbf{u}_1)^2 + \left\|\mathbf{B}_{E} \mathbf{u}_{1}\right\|^2}
{\varphi_{E}(\boldsymbol{\Psi})},
\end{align}
respectively, where
%===================================
\begin{align}~\label{eq:phik} 
& \varphi_k(\boldsymbol{\Psi})=\sum_{t=1}^K\left\|\mathbf{A}_{k, k} \mathbf{u}_{t}\right\|^2+1, \quad k \in \mathcal{K},
\end{align}
%===================================
\begin{align}~\label{eq:phiE} 
&\!\!\!\!\!\!\!\!\!\!\!\!\!\!\!\!\!\!\!\!\!\!\!\!\!\!\varphi_{E}(\boldsymbol{\Psi})=\sum_{t=2}^K\left\|\mathbf{B}_{E} \mathbf{u}_{t}\right\|^2+1.
\end{align}

%=========================================
Now, problem~\eqref{eq:Opt1} can be transformed into a more tractable form as follows
%=====
\begin{subequations}\label{eq:Opt2}
\	\begin{align}
		&\underset{\boldsymbol{\Psi}}{\mathrm{max}}\,\, \hspace{1em}
		 \,\, \frac{\left(\mathbf{a}_1^T \mathbf{u}_1\right)^2}{  \varphi_1(\boldsymbol{\Psi})}
		\\
  &~\text {s.t.} \hspace{1em}  \sum_{k=1}^K \boldsymbol{\Psi}^2(l, k) \leq \rho_{\max }, \quad l \in \mathcal{L}\label{eq:Opt2_c1}\\ 
    &   \hspace{2em} \frac{\left(\mathbf{a}_k^T \mathbf{u}_k\right)^2}{\varphi_k(\boldsymbol{\Psi})} \geq \theta_k ~~ \quad \forall k, k \neq 1, \label{eq:Opt2_c2} \\
        &   \hspace{2em} \frac{\left(\mathbf{b}_{\mathrm{E}} \mathbf{u}_1\right)^2 + \left\|\mathbf{B}_{E} \mathbf{u}_{1}\right\|^2}{\varphi_{\mathrm{E}}(\boldsymbol{\Psi})} \leq \theta_{E}. \label{eq:Opt2_c3}
  	\end{align}
 \end{subequations}
%=======================================
%=======================================
%===========================
%======================
\begin{algorithm}[!t]
\caption{Path-Following Algorithm for Solving~\eqref{eq:Opt3}}
\begin{algorithmic}[1]
\label{alg:Opt}
    \STATE \textbf{Initialize}:  $\kappa=0$, a feasible point $\boldsymbol{\Psi}^{(0)}$ for~\eqref{eq:Opt3}.
\REPEAT
\STATE{Update $\kappa:=\kappa+1$.}
\STATE\text{Solve~\eqref{eq:Opt3} to obtain the optimized solution} $\boldsymbol{\Psi}^{*}$.
\STATE Update $\boldsymbol{\Psi}^{(\kappa)}=\boldsymbol{\Psi}^{*}$ 
\UNTIL{ Convergence.}
\STATE  \text{Return $\boldsymbol{\Psi}^{(\kappa)}$.}
\end{algorithmic}
\end{algorithm}
%======================

Problem~\eqref{eq:Opt2} is a  nonconvex optimization
problem. Therefore, in what follows, we
 use a path-following iterative algorithm~\cite{Nasir:2017:TSP} to solve the  problem efficiently.
The first constraint in~\eqref{eq:Opt2} is obviously convex, while the second constraint  can be written as
%---------
\begin{align} \label{eq:thetak2}
\frac{1}{\sqrt{\theta_k}} \mathbf{a}_k^T \mathbf{u}_k \geq \sqrt{\varphi_k(\boldsymbol{\Psi})}, \quad k \in \mathcal{K} \backslash\{1\},
\end{align}
%---------
which is a second-order cone (SOC) constraint and is convex. Let $\boldsymbol{\Psi}^{(\kappa)}$ denote a feasible point for~\eqref{eq:Opt2} found from the $(\kappa-1)$-th iteration. By invoking the following upper bound
%-----------
\begin{align} \label{eq:inequality}
\frac{x^2}{y} \geq 2 \frac{\bar{x}}{\bar{y}} x-\frac{\bar{x}^2}{\bar{y}^2} y, \quad \forall x>0, y>0, \bar{x}>0, \bar{y}>0,
\end{align}
%-------------
we obtain
 %--------------
\begin{align} \label{eq:thetak3}
\frac{\left(\mathbf{a}_1^T \mathbf{u}_1\right)^2}{\varphi_1(\boldsymbol{\Psi})} \geq f_1^{(\kappa)}(\boldsymbol{\Psi}) \triangleq a^{(\kappa)} \mathbf{a}_1^T \mathbf{u}_1-b^{(\kappa)} \varphi_1(\boldsymbol{\Psi}),
\end{align}
%----------
with
%-----------
\begin{align} \label{eq:akbk}
a^{(\kappa)}=2 \frac{\left(\mathbf{a}_1^T \mathbf{u}_1^{(\kappa)}\right)^2}{\varphi_1\left(\mathbf{\Psi}^{(\kappa)}\right)}, \quad b^{(\kappa)}=\left(a^{(\kappa)} / 2\right)^2.
\end{align}
%------------
Therefore, the objective function  $\left(\mathbf{a}_1^T \mathbf{u}_1\right)^2 / \varphi_1(\boldsymbol{\Psi})$  in~\eqref{eq:Opt2} can be replaced by $f_1^{(\kappa)}(\boldsymbol{\Psi})$. In addition,  since the function $\varphi_{\mathrm{E}}(\boldsymbol{\Psi})$  is convex,  we can use the first-order Taylor approximation of
$\varphi_{\mathrm{E}}(\boldsymbol{\Psi})$ near $\boldsymbol{\Psi}^{(k)}$ to handle the non-convex constraint~\eqref{eq:Opt2_c3}. In particular, constraint~\eqref{eq:Opt2_c3} 
can be  approximated by the convex quadratic constraint 
%====
\begin{align} \label{eq:thetaE2}
\frac{(\mathbf{b}_{\mathrm{E}} \mathbf{u}_1)^2 + \left\|\mathbf{B}_{E} \mathbf{u}_{1}\right\|^2}{ \theta_{E}} \leq \varphi_{E}^{(\kappa)}(\boldsymbol{\Psi}),
\end{align}
%====
with
%-------------
\begin{align} \label{eq:apsiE2}
\varphi_{\mathrm{E}}^{(\kappa)}(\boldsymbol{\Psi}) \triangleq \sum_{k=2}^K\left[\mathbf{u}_k^{(\kappa)^T} \mathbf{B}_E^2\left(2 \mathbf{u}_k-\mathbf{u}_k^{(\kappa)}\right)\right]+1.
\end{align}
%------------
At iteration $(\kappa + 1)$, for a given point $\boldsymbol{\Psi}^{(\kappa)}$, the optimization problem~\eqref{eq:Opt2} can finally be approximated by the following convex problem:
%---------
\begin{subequations}\label{eq:Opt3}
\	\begin{align}
		&\underset{\boldsymbol{\Psi}}{\mathrm{max}}\,\, \hspace{1em}
		 \,\, f_1^{(\kappa)}(\boldsymbol{\Psi}),
		\\
  &~\text {s.t.} \hspace{1em}  \sum_{k=1}^K \boldsymbol{\Psi}^2(l, k) \leq \rho_{\max }, \quad l \in \mathcal{L}\label{eq:Opt3_c1},
  \\ 
    &   \hspace{2em} \frac{1}{\sqrt{\theta_k}} \mathbf{a}_k^T \mathbf{u}_k \geq \sqrt{\varphi_k(\boldsymbol{\Psi})}, \quad k \in \mathcal{K} \backslash\{1\}, \label{eq:Opt3_c2} 
    \\
        &   \hspace{2em} \frac{(\mathbf{b}_{\mathrm{E}} \mathbf{u}_1)^2 + \left\|\mathbf{B}_{E} \mathbf{u}_{1}\right\|^2} {\theta_{E} }\leq \varphi_{E}^{(\kappa)}(\boldsymbol{\Psi}). \label{eq:Opt3_c3}
  	\end{align}
 \end{subequations}
 %-------------

Starting with a feasible point $\boldsymbol{\Psi}^{(0)}$, we solve~\eqref{eq:Opt3} to obtain its optimized solution $\boldsymbol{\Psi}^*$, and use $\boldsymbol{\Psi}^*$
as an initial point in the next iteration. The detailed algorithm for
solving~\eqref{eq:Opt3} is provided in Algorithm~\ref{alg:Opt} where the algorithm terminates when an accuracy level is obtained.
%===================================
\begin{remark}
    Algorithm~\ref{alg:Opt} requires a feasible point $\boldsymbol{\Psi}^{(0)}$ to initialize the procedure. To this end, we first find a feasible point $\boldsymbol{\Psi}^{(0)}$ for the constraints~\eqref{eq:Opt2_c1} and~\eqref{eq:Opt2_c2} and then iteratively solve the  optimization problem  %---------------
\begin{subequations}\label{eq:Opt4}
\	\begin{align}
		&\underset{\boldsymbol{\Psi}}{\mathrm{min}}\,\, \hspace{1em}
		 \,\, \frac{(\mathbf{b}_{\mathrm{E}} \mathbf{u}_1)^2 + \left\|\mathbf{B}_{E} \mathbf{u}_{1}\right\|^2}{ \theta_{E}}-\varphi_{\mathrm{E}}^{(\kappa)}(\boldsymbol{\Psi}), 
		\\
  &~\text {s.t.} \hspace{2em}  \eqref{eq:Opt2_c1}-\eqref{eq:Opt2_c2},
  	\end{align}
 \end{subequations}
%-------------
until the following requirement has been satisfied.
%-----------------
\begin{align} \label{eq:maxapsi2}
\left((\mathbf{b}_{\mathrm{E}} \mathbf{u}_1\right)^2 + \left\|\mathbf{B}_{E} \mathbf{u}_{1}\right\|^2) / \theta_{E}-\varphi_{\mathrm{E}}\left(\boldsymbol{\Psi}^{(\kappa)}\right) \leq 0.
\end{align}
%----------
In this case, $\boldsymbol{\Psi}^{(\kappa)}$ is a feasible solution for Problem~\eqref{eq:Opt3}.\\
\end{remark}
%---------------
\begin{remark}
     Algorithm~\ref{alg:Opt} solves a convex problem at each iteration which involves $A_s=LK$ real-valued scalar variables and $A_q=L+K$ quadratic constraints. Therefore,  the per-iteration complexity of Algorithm~\ref{alg:Opt} is  $\mathcal{O}\big((A_s)^2A_q^{2.5}+A_q^{3.5}\big)$~\cite{Tam:TWC:2017}.
\end{remark}

\section{Detection of Eavesdropping Attack}\label{sec:Eav_detection}
In order to effectively design the AP selection and power allocation as in Section~\ref{sec:APsel_PA}, it is crucial to identify the presence of eavesdroppers in our system and determine which user is being targeted. Thus, it is of importance to have some robust mechanism for detecting  eavesdropping attacks. 
To date, there are a few works on pilot spoofing attack detection. Among them, the authors in~\cite{Qiu:Access:2018} used Gaussian mixture models to identify spoofing attacks, while an additional
random training phase after the conventional
uplink pilot training phase was considered in~\cite{Tian:2017:ACCESS} to identify the attack.
In addition, a two-way training method was
proposed in~\cite{Xiong:TIFS:2016} for attack detection which requires additional downlink training from AP to the legitimate user. The authors in~\cite{Im:TWC:2015} utilized the secret key arrangement protocol  to detect the eavesdropping attack, while the proposed protocol requires several uplink and downlink transmissions among legitimate transmitter and user. A three-step training scheme was proposed in~\cite{Liu:Access:2021} for detecting pilot spoofing attacks in reconfigurable intelligent surface (RIS)-assisted systems, while a two-step training scheme were proposed in~\cite{Xu:icc:2019} and~\cite{Kapetanović:2013:PIMRC} for detection of pilot spoofing attacks in massive MIMO systems, respectively. 
%============
Although the above works make important steps towards
detecting the pilot spoofing attack, they 
investigate simple setups concerning the APs and/or users. More specifically, a
popular assumption in the aforementioned literature is that
there is a single AP and/or a single user in the network. However, in more complicated scenarios such as CF-mMIMO systems, there are multiple cooperating APs  serving multiple users. Therefore, we cannot straightforwardly apply
the proposed detection schemes in the above mentioned works to detect the eavesdropper in our system.}
In \cite{Hoang:TCOM:2018}, the authors proposed a simple method to identify abnormality in pilot training. However, this scheme requires all APs to exchange their received pilot signals, leading to substantial fronthaul demands and increased information exchange overhead. In addition, it is not clear how to obtain the average power of the received pilot signals at the central processing unit (CPU). In this work, we propose a new and efficient method which is also based on the received pilot signals, but can be implemented at each AP. Moreover, the APs do not possess prior knowledge of the users' instantaneous CSI to implement our proposed method.

To identify if user $k$ is overheard by an eavesdropper or not, we consider two hypotheses $\mathcal{H}_{k,0}$ and $\mathcal{H}_{k,1}$, 
 where the latter represents the scenario with an active eavesdropper, while the former represents the scenario without any active eavesdropper. Let $\qy_{l,k}$ denote the projection of the received pilot vector at AP $l$ onto $\Bphik$ (i.e. $\mathbf y_{l, k}= \mathbf Y_{p, l}\Bphik$),  given by
 %===========================
 	\begin{align}~\label{eq:yklm} 
 		\mathbf y_{l, k}= 
 		\begin{cases}\sqrt{\tau_{p} \rho_{u}} \hlmk+ \mathbf{N}_{l}\Bphik, & \mathcal{H}_{k,0}, \\ \sqrt{\tau_{p} \rho_{u}} \hlmk+\sqrt{\tau_{p} \rho_{E}} \hlmE\!+\mathbf{N}_{l}\Bphik, & \mathcal{H}_{k,1}.
 		\end{cases}
 	\end{align}
 %===========================
Then, the average power of the $m$-th element of $\qy_{l,k}$ is given by
 %===========================
 	\begin{align}~\label{eq:power_yklm1} 
 		\Ex \left\{\left| y_{l, k,m} \right|^2\right\}= 
 		\begin{cases}\tau_{p} \rho_{u}\Betalk + 1, & \mathcal{H}_{k,0}, \\ \tau_{p} \rho_{u} \Betalk +\tau_{p} \rho_{E} \betalE\!+1, & \mathcal{H}_{k,1}.
 		\end{cases}
 	\end{align}
 %===========================
From \eqref{eq:power_yklm1}, if AP $l$ knows $\Ex \left\{\left| y_{l, k,m} \right|^2\right\}$, then it can determine if there exists an eavesdropper in the system. More precisely, AP $l$ conducts a comparison between $\Ex \left\{\left| y_{l, k,m} \right|^2\right\}$ and $\tau_{p} \rho_{u}\Betalk + 1$. If these values do not equate, then it indicates the existence of an eavesdropping attempt targeted at user $k$. Otherwise, there is no such eavesdropper. 

In practice, AP $l$ does not  exactly know $\Ex \left\{\left| y_{l, k,m} \right|^2\right\}$. Thus, we propose a simple method to estimate $\Ex \left\{\left| y_{l, k,m} \right|^2\right\}$. The proposed scheme uses the sample average power of $y_{l, k,m}$  as follows:
 %===========================
 	\begin{align}~\label{eq:power_yklm2} 
 		\xi_{l,k,m}= 
 		\frac{\sum_{n=1}^{N_{\text{cb}}} \|\qy_{l,k}(n)\|^2 }{MN_{\text{cb}}},
 	\end{align}
 %===========================
where $N_{\text{cb}}$ is the number of coherence bandwidth intervals within a whole system bandwidth. Since $\qy_{l,k}(n)$ includes i.i.d. elements, $\xi_{l,k,m} \rightarrow \Ex \left\{\left| y_{l, k,m} \right|^2\right\}$ as $MN_{\text{cb}} \to \infty$.

\begin{remark}
    In the 5G NR structure, the system bandwidth is 100 MHz and coherence bandwidth is about 360 KHz. Thus, the number of coherence bandwidth intervals is $N_{\text{cb}}=100\times 10^3/360 = 277$. In addition, the number of antennas per each APs, $M$, typically ranges from 2 to 10. As a consequence, $MN_{\text{cb}}$ is around 544-2770 which is sufficiently large enough to ensure that $\xi_{l,k,m}$ is very close to $ \Ex \left\{\left| y_{l, k,m} \right|^2\right\}$.
\end{remark}
%------------------------------
\begin{figure}[!t]
		%\begin{center}
			\centering 
			\includegraphics[width=0.50\textwidth]{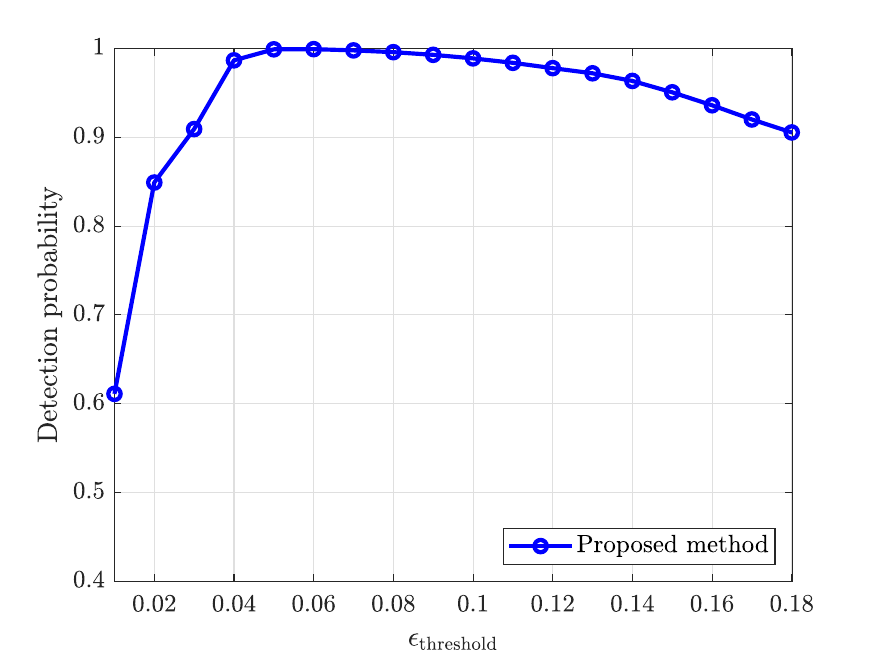}
			\caption{Probability of active eavesdropping attack
detection as a function of $\epsilon_{\text{threshold}}$. }
			\label{fig2}
			%\end{center}
\end{figure}
%==================
By using \eqref{eq:power_yklm2}, we can estimate $\Ex \left\{\left| y_{l, k,m} \right|^2\right\}$. Intuitively, if $\xi_{l,k,m}$ is close to $\tau_{p} \rho_{u}\Betalk + 1$, then we can conclude that there is no eavesdropper attacking user $k$. Building upon this insight, we propose a simple method to identify the presence of eavesdroppers in our system and determine which user is being targeted, outlined as follows:
\begin{itemize}
    \item \textbf{Step 1}: At AP $l$, compute $\xi_{l,k,m}$ using \eqref{eq:power_yklm2}.
    \item \textbf{Step 2}: Compute the ratio $$\upsilon_{l,k,m}=\xi_{l,k,m}/\left(\tau_{p} \rho_{u}\Betalk + 1\right),$$ and choose a threshold $\epsilon_{\text{threshold}}$. If $\left|\upsilon_{l,k,m}-1\right| \leq \epsilon_{\text{threshold}}$, then  AP $l$ decides that there is no eavesdropper attacking user $k$. Otherwise, user $k$ is overheard by an eavesdropper. 
   \item \textbf{Step 3}: If a majority of the APs determine that user $k$ is being overheard by an eavesdropper, then we conclude that the system is overheard by an eavesdropper, and user $k$ is  specifically targeted. Otherwise, we can conclude that there is no eavesdropper present in the system attacking user $k$.  
\end{itemize}

Figure~\ref{fig2} shows the detection probability  of an active eavesdropping attack relying on our proposed method for various  values of $\epsilon_{\text{threshold}}$. The simulation settings  resemble those explained in Section~\ref{Section:Simulation} with $L=10$,  $M=2$, and $K=4$.
It is seen that the proposed method works very well (i.e., the detection probability is close to $1$). This implies that our  assumption that the system knows the presence of an eavesdropper is reasonable. 
In addition, we observe that a trade-off exists between $\epsilon_{\text{threshold}}$ and the detection probability. More specifically, first, as $\epsilon_{\text{threshold}}$ increases, the detection probability increases but
it then starts decreasing as $\epsilon_{\text{threshold}}$ increases beyond the optimized
value. The intuitive reason is that a large $\epsilon_{\text{threshold}}$ can lead to misidentifying a user as attacked due to the fact that for the majority of the APs we would have $\left|\upsilon_{l,k,m}-1\right| \leq \epsilon_{\text{threshold}}$, while a small $\epsilon_{\text{threshold}}$ may result in falsely dismissing the presence of an eavesdropper. Therefore, it is important to optimize $\epsilon_{\text{threshold}}$ to
maximize the detection probability. For this network setup, its turns out that the optimized values for  $\epsilon_{\text{threshold}}$ are between $0.05$ to $0.08$.

%================================
\section{Simulation Results}\label{Section:Simulation}
%%%%%%%%%%%%%%%%%%%%%%%%%%%%%%%%%%%%%%%%%%%%%
In this section, we provide numerical results to study the secrecy performance of CF-mMIMO with PPZF  precoding in scenarios involving an active eavesdropper, as well as to verify the benefit of our AP selection and power allocation schemes. We also include MRT precoding as a benchmark for comparison.
The APs and the users are randomly located within a square  of size $1 \times 1$ km$^2$, which is wrapped around at the edges to avoid the boundary effects.  Moreover, the eavesdropper is randomly located inside a circle with radius $r$ around  user $1$. The  large-scale fading is modelled as 
%------
\begin{align}
\beta_{l,k} = 10^{\frac{\text{PL}_{l,k}^d}{10}}10^{\frac{F_{l,k}}{10}},
\end{align}
%---------
where $10^{\frac{\text{PL}_{l,k}^d}{10}}$ denotes the path loss, and $10^{\frac{F_{l,k}}{10}}$ depicts the shadowing effect with $F_{l,k}\in\mathcal{N}(0,4^2)$ (in dB).  Moreover, $\text{PL}_{l,k}^d$ (in dB) is given by
%---------------------------------------------------
\begin{align}
\label{PL:model}
\text{PL}_{l,k}^d = -30.5-36.7\log_{10}\Big(\frac{d_{l,k}}{1\,\text{m}}\Big),
\end{align}
%---------------------------------------------------
and the correlation among the shadowing terms from the AP $l$ to different users $k$ is given by  
%------
\begin{align}
	 \mathbb{E}\{F_{l,k}F_{j,k'}\} \triangleq \begin{cases} 4^22^{-\zeta_{k,k'}/9\,\text{m}}, & j = l, \\ 0, & \text{otherwise}, \end{cases}
\end{align}
%-----------------------------------------------------
where $\zeta_{k,k'}$ is the physical distance between users $k$ and $k'$ \cite{Emil:TWC:2020}. 
In addition, we choose the bandwidth $B=20$ MHz,  a noise power equal to  $-92$ dBm, and the maximum
transmit power for each AP and each user  $200$ mW  and $100$ mW, respectively. Also, $\rho_E$ = $\rho_u$ and $\tau_{p}=K$. In all figures, we show the average SSE, while the average is taken over the large-scale fading.

%===========================
\begin{figure}[!t]
 		%\begin{center}
 			\centering 
 			\includegraphics[width=0.50\textwidth]{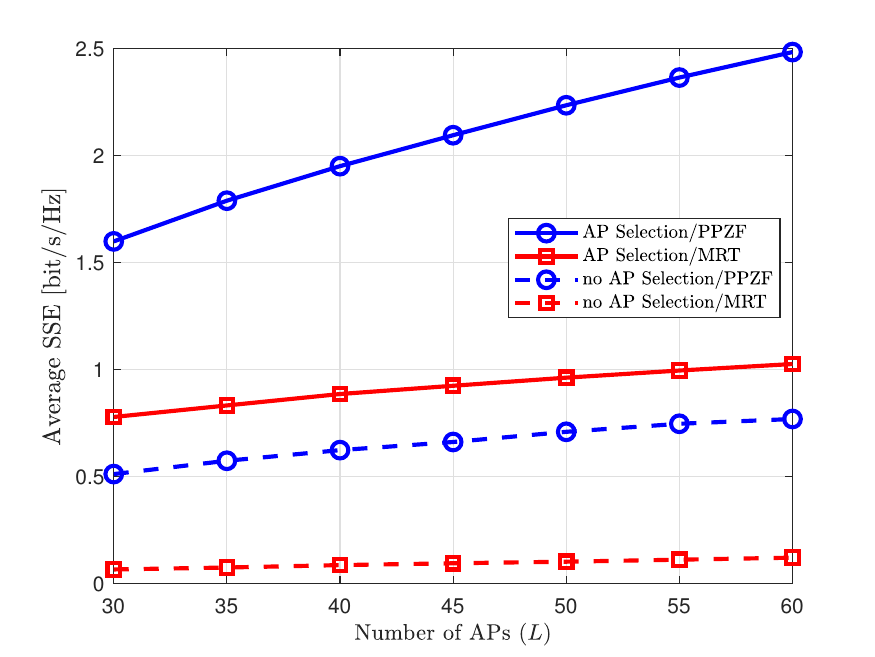}
 			\caption{SSE versus $L$ for PZZF and MRT precoding schemes with AP selection and without it. Here, $M=4$, $K=10$, and $r = 100$ m.}
 			\label{fig3}
 			%\end{center}
\end{figure}
%===========================

%===========================
% \begin{figure}[!t]
% 		%\begin{center}
% 			\centering 
% 			\includegraphics[width=0.50\textwidth]{./Image/Fig3revised.eps}
% 			\caption{SSE versus $L$ for PZZF, ZF, and MRT precoding schemes with AP selection and without it. Here, $M=4$, $K=10$, and $r = 100$ m.}
% 			\label{fig2}
% 			%\end{center}
% \end{figure}
%===========================

\subsection{Performance Evaluation}
% Figure~\ref{fig1} compares the behavior of CF-mMIMO system utilizing all APs and APs selected using the heuristic AP selection schemes by showing their CDFs of the SE of user 1 and eavesdropper under PPZF and MRT precoding designs. It is observed that the heuristic AP selection scheme decreases the SE of eavesdropper significantly even though it also decreases that of user $1$ with a lower rate. Moreover, PPZF schemes overperform MRT ones in this regard due to the capability of PPZF schemes to mitigate inter-user interference.

\emph {1) Performance of the Proposed AP Selection}: 
First, we investigate the performance of the proposed AP selection in Algorithm~\ref{alg:GR} for different  precoding schemes.  Figure~\ref{fig3} shows the SSE versus the number of APs, $L$. From this figure, we can draw
the following conclusions:
\begin{itemize}
\item The proposed AP selection provides a noticeable secrecy improvement. More specifically, it provides performance gains of up to $220\%$ and $730\%$
for the CF-mMIMO system with PPZF and MRT precoding schemes, respectively. This performance gain is reasonable:   i) primarily due to the fact that in a CF-mMIMO system with distributed APs, there are some APs which are in the
close vicinity of the eavesdropper and hence serving the attacked user $1$ by these APs
may result in high values for the overheard SINR; ii) secondly, there are some APs which are located far away from user $1$ and, thus, will not contribute significantly to its overall SE.
\item The PPZF scheme provides a better SSE performance
than the MRT scheme due to its ability to cancel the inter-user interference.
This result highlights the importance of precoding scheme for secure CF-mMIMO systems. 
\item The gap between PPZF and MRT is always noticeable. Interestingly, PPZF along with the proposed AP selection provides around $2$-fold improvement in the SSE compared to the MRT scheme for scenarios with a varying number of APs. In principle, this is reasonable since compared to MRT, PPZF offers an excellent balance between interference cancellation and the boosting of the desired signal. More importantly, it can provide high interference cancellation over a wide range of the numbers of antennas/active APs. Therefore, with PPZF, the benefits of the AP selection to enhance the SSE are more pronounced.
\end{itemize}
%============
\begin{figure}[!t]
		%\begin{center}
			\centering 
			\includegraphics[width=0.50\textwidth]{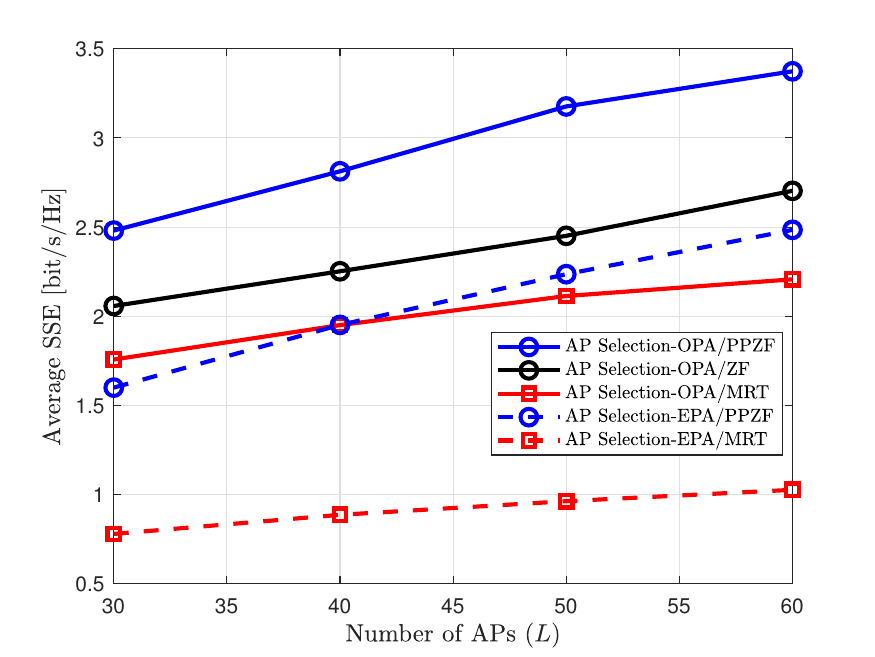}
			\caption{SSE with AP selection and the proposed power allocation versus $L$ for PPZF, ZF, and MRT precoding schemes. Here, $M=4$ and $K=10$, and $r = 100$ m with AP selection scheme, EPA, and OPA.}
			\label{fig4}
			%\end{center}
\end{figure}
%===========================

% \com{The improvement comes from the fact that the increase in APs increases the macro-diversity and decreases the path losses. However, it causes high interference which can be mitigated by PPZF more efficiently than MRT. Consequently, PPZF benefits from increasing APs more than MRT.}
% \item Figure~\ref{fig2} also validates the accuracy of the derived SSE expression for the PPZF scheme in~\eqref{eq:R_sec} using the SINR expressions of legitimate user and eavesdropper in~\eqref{eq:SINR_k2} and~\eqref{eq:SINR_E2}, respectively. It is observed that the obtained results in closed-form  (solid curves) and those obtained by Monte Carlo simulations (markers) are perfectly matched.
% \item \com{The presence of an AP near the Eve contributes to increasing the leaked information of the Eve and thus reduces the security SE, especially if the AP is far from the attacked user at the same time. Therefore, taking this AP out of service using the proposed AP selection approach improves the secrecy performance significantly.
% The improvement in secrecy performance with MRT is more than with PPZF because the decrease in APs reduces the interference which in turn enhances the performance of MRT more significantly than PPZF.}
% \end{itemize}
 %==================================
\begin{figure}[t]
\includegraphics[width=0.50\textwidth] {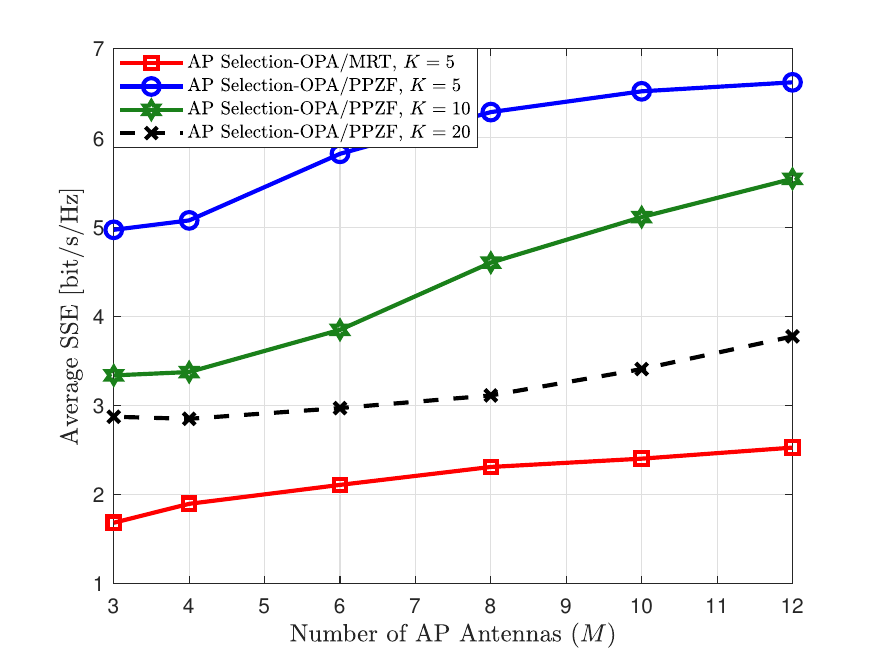}
\caption{SSE with AP selection and the proposed power allocation versus $M$, when $LM=240$ and $r = 100$ m.}\label{fig5}
\end{figure}

\emph {2) Performance of the Proposed Power Optimization}:
Figure~\ref{fig4} shows the advantage of the proposed power allocation as a function of the number of APs for CF-mMIMO system with the proposed AP selection. In this figure, 
``EPA'' corresponds to the case with equal power allocation, i.e., $\rholk=\rho_{\text{max}}/K$ for all $l$ and $k$, while ``OPA'' corresponds to our proposed optimized power allocation algorithm (i.e. Algorithm~\ref{alg:Opt}). Numerical results lead to
the following observations:
\begin{itemize}
\item
 The CF-mMIMO system with the proposed power allocation remarkably outperforms the CF-mMIMO system with the EPA scheme. More specifically, the proposed power allocation in Algorithm~\ref{alg:Opt} provides additional performance gain of around $55 \%$ and $100 \%$
over equal power allocation for a CF-mMIMO system relying on    PPZF  and MRT precoding scheme, respectively, which demonstrates the significance of our power allocation solution.
\item 
Increasing the number of APs along with the proposed power allocation scheme increases the gain of PPZF over MRT. 
\item In Fig.~\ref{fig4}, we also present the SSE performance of the CF-mMIMO system with ZF precoding, where the interference towards all the users is suppressed.  Compared to ZF design,  PPZF offers better SSE performance (around $30$\%). This is reasonable, since ZF avails of a modest array gain of $M-K$ as almost all the DoFs are used to mitigate the interference.
\end{itemize}
% Fig.~\ref{fig3} presents the secrecy rate as a function of the number of antennas (APs) using the proposed AP selection schemes with equal and sub-optimal power allocation coefficients and PPZF and MRT precoders. It is obvious that another enhancement rate has been added to the secrecy performance by the suggested optimization problem especially using PPZF precoding schemes.  

% Figure~\ref{fig4} shows the SSE of CF-mMIMO system under active eavesdropping with PPZF as a function of the number of AP antennas $M$ and the radius of a circle around user $1$ inside it eavesdropper is randomly located ($r$). It can be clearly seen that the secrecy rate improves with the increase of $r$ and $M$. However, by comparing this figure with Figure~\ref{fig3}, we can see that the secrecy enhancement obtained by increasing $M$ is not as high as that obtained by increasing the number of APs even though increasing the number of APs leads to an increase in interference with a rate larger than that caused by increasing the number of antennas $M$. This may be attributed to the high capability of PPZF schemes to suppress interference caused due to using large number of APs.
%===========================================
\begin{figure}[t]
		%\begin{center}
			\centering 
			\includegraphics[width=0.50\textwidth]{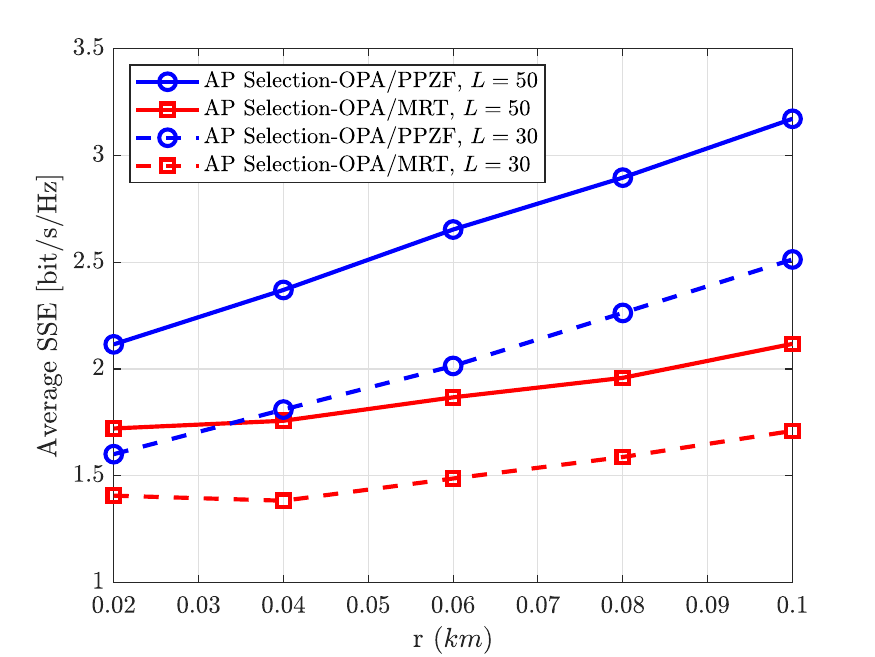}
			\caption{SSE with AP selection and optimized power allocation versus $r$ for $L=30$ and $L=50$, and $K=20$.}
			\label{fig6}
			%\end{center}
\end{figure}
%=============================================
%===========================================
\begin{figure}[t]
		%\begin{center}
			\centering 
			\includegraphics[width=0.50\textwidth]{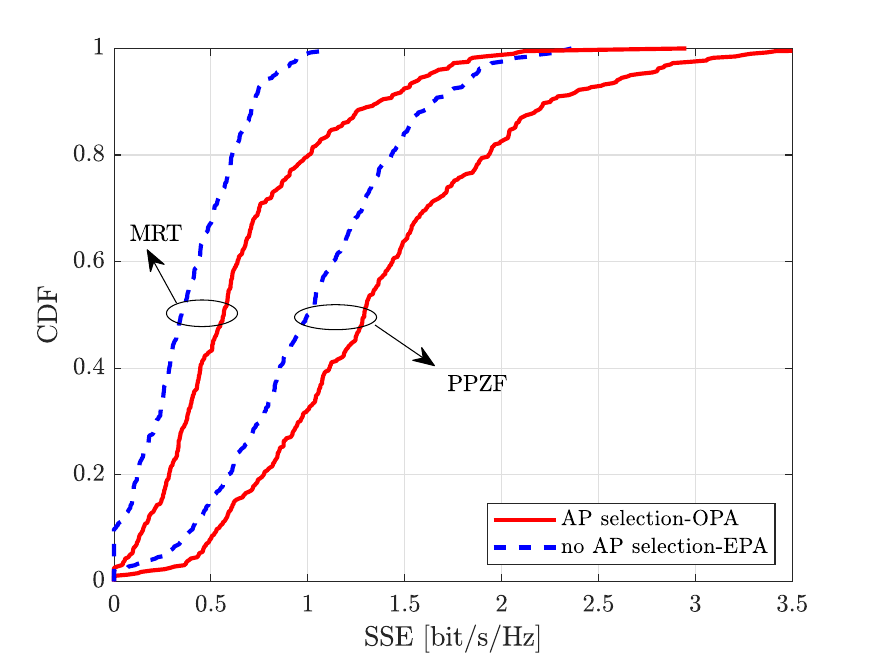}
			\caption{CDF of the SSE with AP selection and optimized power allocation for the QuaDRiGa  channel model, $L=30$, $K=10$, and $r=400$ m.}
			\label{figquad}
			%\end{center}
\end{figure}
%=============================================
We would like to highlight that the proposed AP selection and power optimization algorithms can be applied to different channel models. More precisely, we can use the analytical expressions in~\eqref{eq:SINRk} and~\eqref{I_E}, instead of the closed-form expressions in~\eqref{eq:SINR_k2} and~\eqref{eq:SINR_E2}, respectively, to numerically implement our proposed AP selection Algorithm 1 and power optimization Algorithm 2.
For example, by using  channel realizations from the $3$GPP   $38.901$  channel model  implemented in QuaDRiGa~\cite{QUADRIGA},  developed by the Fraunhofer Heinrich Hertz Institute, we can obtain the result in Fig.~\ref{figquad}, which verifies the benefit of our proposed scheme. More specifically, the proposed AP selection and power allocation yield a median SSE performance gain   of around $65\%$ and $30\%$ for the CF-mMIMO system with the QuaDRiGa channel model relying on PMRT and PPZF precoding designs, respectively.

%===================
\emph{3) Effect of the Number of Antennas Per AP:}
In Fig.~\ref{fig5}, we examine the SSE of CF-mMIMO as a function of the number of antennas per AP, $M$, when the total number of service antennas is fixed, i.e., $LM=240$. For a fixed $LM$,  increasing the number of antennas per AP (reducing the number of APs) has two effects on the SSE, namely,  (i) higher array gain (a positive effect) and (ii) lower degree of macro-diversity and higher path losses (a negative effect).
The positive effect becomes dominant,  which enhances  the SSE performance. In addition, we observe that by
increasing the number of users, $K$, the secrecy performance of the CF-mMIMO system deteriorates. Nevertheless, the CF-mMIMO system comprising 
a high number of users and employing the proposed AP selection and power optimization still yields an excellent SSE performance. Interestingly, the CF-mMIMO system with $20$ users and PPZF scheme offers around $60 \%$ SSE gain compared to the system comprising $5$ users with MRT scheme. This is reasonable because, PPZF  has the ability to cancel the inter-user interference.

% However, for PPZF precoding design if the ratio $K/M$ is small, the latter effect becomes dominant for higher values of $M$ which leads to an improvement in the SSE performance. This is the reason for why in Fig.~\ref{fig5} when $M$ increases and $K=5$, the SSE first decreases and then increases. This is reasonable because  the use of PPZF gives the array gain $M-|\mathcal{S}_l|$ which depends on the number of users. In particular, PPZF provides a trade-off between
% inter-user interference mitigation and boosting of the desired signal, whose
%  operation depends on the $K/M$. On the other hand,  MRT cannot cancel interference by nature and provides the array gain $M$, which is independent of the number of users. 

\emph{4) Effect of the Eavesdropper Location:}
Figure~\ref{fig6} shows the SSE of CF-mMIMO system as a function of  the position of the eavesdropper, the radius of a circle around user 1, $r$. As expected, the SSE  improves with the increase of $r$.
 We can also see that the secrecy enhancement obtained by using a higher number of APs is more pronounced when the eavesdropper is located closer to the legitimate user $1$.
%%%%%%%%%%%%%%%%%%%%%%%%%%%%%%%%%%%%%%%%%%%%%%%%%%%%%%%%%%%%%%%%%%%%%
\subsection{Multiple-Antenna Eavesdropper}
 In the previous sections, we studied  the performance of secure CF-mMIMO under active eavesdropping when the eavesdropper has a single antenna. We now examine how the SSE performance changes as eavesdropper is equipped with multiple antennas.
%===========
We consider an eavesdropper with $N_E$ antennas  and  then rewrite the eavesdropper’s receive signals in~\eqref{z_E2} as follows:
%=============
\begin{align}\label{z_E3} 
\qz_{E} &= \sum\limits_{l=1}^{L} \!\!
\sqrt{\rho_{l, 1}}\mathbf H_{l,E}^{\text{H}}  \qw_{l,1} s_{1} 
\!+ \!\! {\sum\limits_{\substack{ t \neq 1}}^{K} } \! 
\sum\limits_{l=1}^{L} \!\!
\sqrt{\rholt} \mathbf H_{l,E}^{\mathrm{H}}  \qw_{l,t}s_{t}\! +\! \qn_{E},
\end{align}
%=============
where $\qz_E \in \mathcal{C} ^{N_E \times 1}$
represents the eavesdropper's received signals across her $N_E$ antennas, $\qH_{l,E} \in \mathcal{C}^ { M \times N_E}$ is the channel matrix from the $l$-th AP  to the eavesdropper, and $\qn_E \in \mathcal{C}^ {M \times N_E}$
is the additive zero-mean Gaussian noise vector each with variance equal to one. To detect the signal transmitted to user 1, the most effective and low-complexity strategy for the eavesdropper is to employ the maximum-ratio combining (MRC) scheme.
% Given that the effective channel from the AP $l$ to eavesdropper is  $ve = Hew$, eavesdropper can utilize MRC only when she possesses knowledge of her own channel matrix $He$ as well as the precoding vectors $\qw_{l,1}$ employed by the  APs. It is important to note that, for the purpose of the worst-case analysis, we are assuming these strong assumptions.
Therefore, the SE at the eavesdropper can be written as
%===========
\begin{align} 
R_E^{\mathrm{MRC}} &=\log _2\left(1+\SINRE^{\mathrm{MRC}}\right),
\end{align}
%===========
%============
\begin{figure}[!t]
		%\begin{center}
			\centering 
			\includegraphics[width=0.50\textwidth]{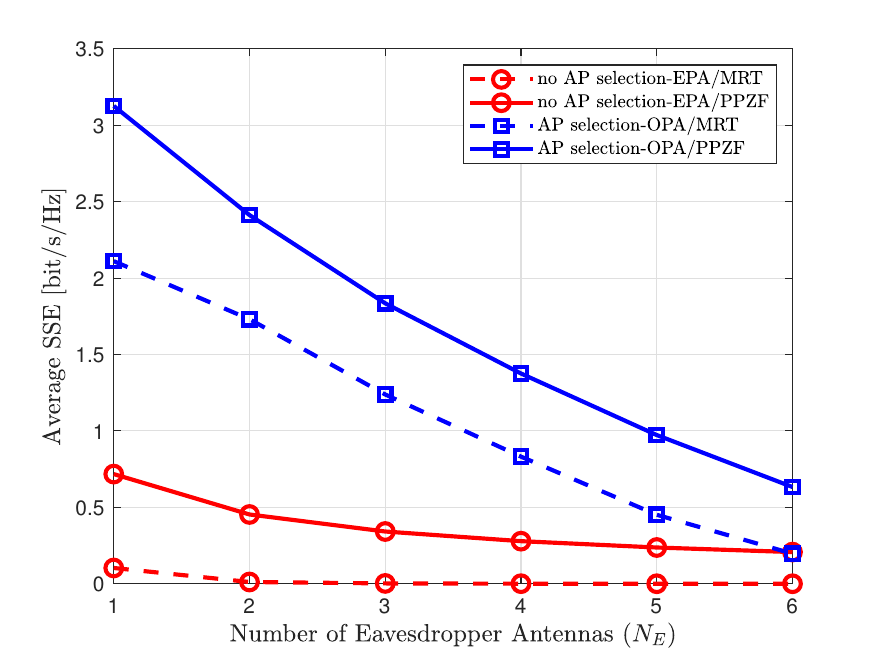}
			\caption{SSE with AP selection and the proposed power allocation versus number of eavesdropper antennas $N_E$ for PPZF and MRT precoding schemes. Here, $L=50$, $M=4$, $K=10$, and $r = 100$ m.}
			\label{fig7}
			%\end{center}
\end{figure}
%===========================
where
%=============
\begin{align} \label{I_E_2}
\SINRE^{\mathrm{MRC}}\approx\sum_{n=1}^{N_E} \frac {\mathbb{E} \{ |\BUE^n|^{2} \}}{{ \sum\limits_{\substack{ t \neq 1}}^{K} }{{  \mathbb{E} \{|\mathrm {UI}^n_{ \textrm {E}, t}|^{2}}} \}+ 1 }=\sum_{n=1}^{N_E}\SINRE^n,
 \end{align}
%=========
where $\SINRE^n$ is the received SINR  at the $n$-th antenna of the eavesdropper given in~\eqref{eq:SINR_E2}, while Eq.~\eqref{I_E_2} follows from the well known MRC result that the collective SNR is the summation of SNRs at each element.

Figure~\ref{fig7} shows the SSE as a function of the number of antennas at eavesdropper, $N_E$. The findings show that increasing the number of eavesdropper antennas has a significant impact on the secrecy performance of CF-mMIMO for both PPZF and MRT precoding schemes. Furthermore, we can infer that the reduction in the SSE attributed to the increased number of eavesdropper antennas is more pronounced when employing
our proposed AP selection and power optimization approaches compared to the scenario where all APs are utilized with equal power allocation.
%============
\begin{figure}[!t]
		%\begin{center}
			\centering 
			\includegraphics[width=0.50\textwidth]{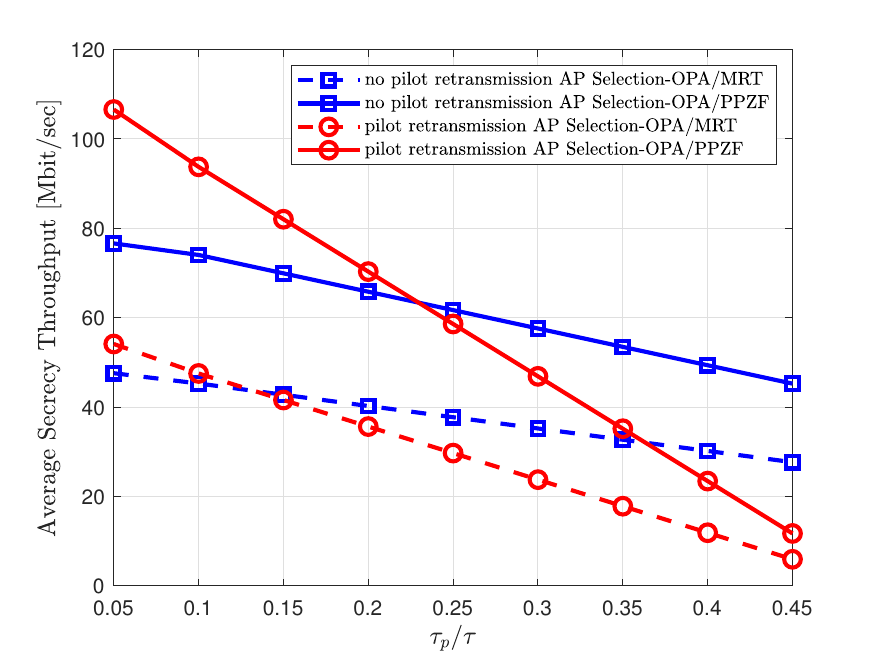}
			\caption{Secrecy throughput with AP selection and the proposed power allocation versus  $\tau_p/\tau$ for PPZF and MRT precoding schemes. Here, $L=50$, $M=4$, $K=10$, $B=20$ MHz, and $r = 200$ m.}
			\label{fig8}
			%\end{center}
\end{figure}
%===========================
%==============
\subsection{Mitigation of Eavesdropping Attack}
As discussed in Section~\ref{sec:sysmodel}, the pilot spoofing  attack during
the uplink training phase degrades considerably the SSE performance of the system.
Therefore, in this subsection we develop a counter strategy to suppress the effect of pilot spoofing attack. We propose a \emph{pilot re-transmission scheme} where the pilot sequences will be re-transmitted when the eavesdropping attack has been detected using the proposed scheme in Section~\ref{sec:Eav_detection}.
We consider the worst case where the eavesdropper has prior knowledge of the  pilot sequences and tries to launch a pilot spoofing attack using  those training sequences. In this case, the CF-mMIMO system can adapt the training sequences based on the knowledge of the current pilot transmission rather than merely transmitting them randomly.
Our proposed pilot re-transmission scheme is outlined as follows:
\begin{itemize}
\item \textbf{Initialization}: Choose the value of pilot length $\tau_p$ and set the number of the pilot re-transmission rounds as $N_{\mathrm{pr}}=1$.
\item \textbf{Step 1}: Each user $k$ sends pilot sequence $\Bphik$.
\item \textbf{Step 2}: The eavesdropping attack detection scheme in Section~\ref{sec:Eav_detection} is performed by the APs for all the users. 
\item \textbf{Step 3}: If there is no eavesdropper attacking Stop. Otherwise, go to step 4.
\item \textbf{Step 4}: If we conclude that the system is overheard by an eavesdropper, and one user (say user $1$) is specifically targeted, then the system will find $\Bphione^*$ such that the eavesdropping attack detection scheme results that there is no eavesdropper present in the system attacking user $1$. Then, user $1$ will re-transmit this new pilot  $\Bphione^*$ and other users will re-transmit the new pilots from the remaining  orthogonal pilot set. Set $N_{\mathrm{pr}}=2$. 
\end{itemize}
Accordingly, the achievable SE expression for user $1$  can be obtained as
%=======
\begin{align} \label{eq:pr_SE_k1} 
R_1^{\mathrm{pr}}= \log_2\left(1+\SINRone^{\mathrm{pr}} \right),
\end{align}%====
where
%==========
\begin{align}~\label{eq:SINR_k_pr} 
\SINRone^{\mathrm{pr}}= \frac{\left(\!\!\!\!\quad\sum\limits_{l=1}^{L} \sqrt{\left(M-\tausl\right) \rho_{l,1} \gamma_{l, 1}^{\mathrm{pr}}}\right)^{2}}{\sum\limits_{t=1}^{K} \sum\limits_{l=1}^{L} \rholt\left(\beta_{l, 1}-\delta _{l,1}\gamma_{l, 1}^{\mathrm{pr}}\right)+1},
\end{align}
%===========================
with 
\begin{align}~\label{eq:gamalmk_pr}  
		\gamma_{l,1}^{\mathrm{pr}}= 
  \begin{cases}
  \frac{\tau_{p} \rho_{u} \betalone^{2}}{\tau_{p} \rho_{u} \betalone+1}, & N_{\mathrm{pr}}= 2, 
  \\ 
  \frac{\tau_{p} \rho_{u} \betalone^{2}}{\tau_{p} \rho_{u} \betalone+1}, 
  & N_{\mathrm{pr}}=1~\text{and}~\mathcal{H}_{1,0},
    \\ 
  \frac{\tau_{p} \rho_{u} \betalone^{2}}{\tau_{p} \rho_{u} \betalone+\tau_{p} \rho_{E} \betalE+1}, 
  & N_{\mathrm{pr}}=1~\text{and}~\mathcal{H}_{1,1},
		\end{cases} 
\end{align}
%======
and the achievable SE of the eavesdropper can be obtained as $R_E^{\mathrm{pr}}= \log_2 (1+\SINRE^{\mathrm{pr}})$ where
%==========
\begin{align} 
\SINRE^{\mathrm{pr}}=\frac { \sum\limits_{l=1}^{L}\rho_{l,1}\beta_{l,E}}{{\sum\limits_{\substack{ t \neq 1}}^{K}} \sum\limits_{l=1}^{L}\rholt\beta_{l,E}+1}. 
\end{align}
%==============
We highlight that the proposed pilot re-transmission scheme might consume more uplink training resources. Therefore, the cost of the pilot re-transmission scheme entails a loss in the SE of the users which highly depends on the application scenarios. For example, this cost would be non-negligible when the length of the coherence interval is small (e.g., in high mobility environments or/and the number of users is large) or/and the length of the uplink training phase is large. Therefore, to ensure a fair comparison between conventional CF-mMIMO (with no pilot re-transmission) and secure CF-mMIMO with pilot re-transmission scheme, we consider the per-user (eavesdropper)  throughputs which take into account the pilot training overhead and are defined as
%=======
\begin{align}
S_a=B\left(1-\frac{\tau_p}{\tau}\right) {R_a},
\end{align}
%================
and
%=======
\begin{align}
S_a^{\mathrm{pr}}=B\left(1-\frac{N_{\mathrm{pr}}\tau_p}{\tau}\right)R_a^{\mathrm{pr}},
\end{align}
%================
where $a \in \{1,E\}$, $\tau$ is the total length of the coherence interval in samples, and $B$ is the spectral bandwidth. The term $\tau_p/\tau$  signifies the fact that, for each coherence interval of length $\tau$ samples and with no pilot re-transmission, we consume $\tau_p$ samples for the uplink  training phase, while for the case of pilot re-transmission, we consume $N_{\mathrm{pr}}\tau_p$ samples.

Figure \ref{fig8} compares the performance of the CF-mMIMO with pilot re-transmission and with no pilot re-transmission schemes as a function of   $\tau_p / \tau$ for PPZF and MRT precoding designs. The results indicate that for PPZF precoding, pilot re-transmission outperforms no pilot re-transmission when
$\tau_p / \tau  < 0.23$, and exhibits an inferior performance when $\tau_p / \tau  \geq 0.23$. This result is in accordance with
the above discussion and implies that the lengths of  $\tau$ and $\tau_p$  are the key factors determining the
extent to which CF-mMIMO with pilot re-transmission outperforms conventional CF-mMIMO (with no pilot re-transmission).

% Specifically, when $\tau_p / \tau = 0.05$, the SSE shows an increase of approximately 39$\%$ and 14$\%$ for PPZF and MRT precoding schemes, respectively. Nonetheless, it experiences a rapid decline as the coherence interval diminishes (i.e., when $\tau_p / \tau$ increases).
%%%%%%%%%%%%%%%%%%%%%%%%%%%%%%%%%%%%%%%%%%%%%%%%%%%%%%%
\section {Conclusion}
%%%%%%%%%%%%%%%%%%%%%%%%%%%%%%%%%%%%%%%%%%%%%%%%%%%%%%%
We investigated the SSE of  CF-mMIMO  systems with multi-antenna APs and PPZF precoding under an active eavesdropping attack and the presence of channel estimation errors.  We derived closed-form expressions for the  SE at the legitimate user and eavesdropper, and thereby the SSE. 
We proposed a large-scale-based AP selection approach to improve the SSE and  a large-scale-based power optimization approach  to maximize the received SINR at the
legitimate user, subject to a maximum allowable SINR at the eavesdropper
and  individual QoS  requirements and transmit power constraints. We showed that our proposed AP selection and power optimization approaches provide significant SE gains. Our results also confirm that for a CF-mMIMO system experiencing active eavesdropping, the PPZF precoding scheme can achieve a noticeable SSE gain compared to the MRT precoding scheme.  Insights from the simulation results show that the number of APs is a dominating factor of the SSE performance. When the number of APs is large, a secure  CF-mMIMO system relying on the proposed approaches and PPZF scheme can offer excellent performance even for higher user loads. Note that the proposed AP selection and power optimization algorithms are not limited to Rayleigh fading channels. They are applicable to various channel models by using numerical methods, as evidenced in Fig.~\ref{figquad}. Finally, investigating CF-mMIMO systems under multiple-eavesdropper attacks and developing joint AP selection and power optimization schemes to maximize the SSE are interesting research topics for future research.
%%%%%%%%%%%%%%%%%%%%%%%%%%%%%%%%%%%%%%%%%%%%%%%%%%%%%%%
\appendices
%%%%%%%%%%%%%%%%%%%%%%%%%%%%%%%%%%%%%%%%%%%%%%%%%%%%%%%
\section{Proof of Proposition~\ref{prop:SINRk}}~\label{app:SINRK} 
We first calculate the numerator of \eqref{eq:SINRk} as
\begin{align}~\label{eq:ExBUK} &\Bigg|\sum \limits _{l \in \mathcal {Z} _{k}} \sqrt {\rho _{l,k}} \Ex\left \{{{(\hat { \mathbf {h}}_{l,k}+\tilde { \mathbf {h}}_{l,k}) ^{\text {H}} \mathbf {w} ^ {\textrm {PZF}} _{l, {k}}}}\right \}\notag\\&+\sum \limits _{p \in \mathcal {M} _{k}} \sqrt {\rho _{p,k}} \Ex\left \{{{ (\hat { \mathbf {h}}_{p,k}+\tilde { \mathbf {h}}_{p,k}) ^{\text {H}} \mathbf {w} ^ {\textrm {PMRT}} _{p,j}}}\right \}\Bigg |^{2}\notag\\
&=\left ({\sum \limits _{l \in \mathcal {Z} _{k}} \textbf\!\!\sqrt {(M\!\!-\!\! \tausl)\rho _{l,k}\gamma _{l,k}}\!  + \!\!\!\sum \limits _{p \in \mathcal {M} _{k}} \!\!\sqrt {(M\!\!- \!\!\tausp)\rho _{p,k}\gamma _{p,k}} }\right)^{2},\notag\\
&=\left ({\sum \limits _{l=1}^{L} \sqrt {(M\!-\!\tausl)\rho _{l,k}\gamma _{l,k}}}\right)^{2}.\end{align}
The first term in the denominator of \eqref{eq:SINRk} is given by
\begin{align}~\label{eq:ExBUK2}
&\sum_{t=1}^{K}\! \Ex \!\left\{\left|\sum_{l \in \mathcal{Z}_{t}} \!\!\sqrt{\rho_{l, t}} \mathbf{h}_{l, k}^{\text{H}} \mathbf{w}_{l, {t}}^{\textrm{PZF}}\!+\!\sum_{p \in \mathcal{M}_{t}} \!\!\sqrt{\rho_{p, l}} \mathbf{h}_{p, k}^{\text{H}} \mathbf{w}_{p, {t}}^{\textrm{PMRT}}\right|^{2}\right\}\!\notag\\&=\!\sum_{t=1}^{K}\! \Ex \!\left\{\left|\sum_{l \in \mathcal{Z}_{t}} \!\!\sqrt{\rho_{l, t}} \mathbf{h}_{l, k}^{\text{H}} \mathbf{w}_{l, {t}}^{\textrm{PZF}}\right|^{2}\right\}\!\nonumber\\
&+\!2 \!\sum_{t=1}^{K}\! \operatorname{Re}\!\left\{\!\sum_{l \in \mathcal{Z}_{t}} \sum_{p \in \mathcal{M}_{t}}\!\! \!\sqrt{\rho_{l, t} \rho_{p, t}} \Ex \!\left\{\mathbf{h}_{l, k}^{\text{H}} \mathbf{w}_{l, {t}}^{\textrm{PZF}}\left(\mathbf{w}_{p, {t}}^{\textrm{PMRT}}\right)^{\text{H}}\! \mathbf{h}_{p, k}\right\}\!\!\right\}\notag\\&+\!\!\sum_{t=1}^{K}\! \Ex \left\{\left|\sum_{p \in \mathcal{M}_{t}} \!\!\sqrt{\rho_{p, t}} \mathbf{h}_{p, k}^{\text{H}} \mathbf{w}_{p, t}^{\textrm{PMRT}}\right|^{2}\right\}.
\end{align}
%--------------
Now, we compute  the first term of \eqref{eq:ExBUK2} as follows:

If $k \in \mathcal{S}_{l}$ , then the term $\mathbf{h}_{l, k}^{\text{H}} \mathbf{w}_{l, t}^{\textrm{PZF}}$ can be calculated based on \eqref{eq:alphaPZF}. Thus, we have
%--------------
\begin{align}\label{eq:ExBUK3}&\sum \limits _{t=1}^{K} \mathop {\mathrm {\Ex }}\nolimits \left \{{\left |{ \sum \limits _{l \in \mathcal {Z} _{t}}\sqrt {\rho_{l,t}} \mathbf {h}_{l,k} ^{\text{H}} \mathbf {w}^{\textrm{PZF}} _{l,{t}}}\right |^{2}}\right \}\notag\\&=\sum \limits _{t=1}^{K} \mathop {\mathrm {\Ex }}\nolimits \left \{{\left |{ \sum \limits _{l \in \mathcal {Z} _{t}}\sqrt {\rho _{l,t}}(\alpha ^ {\textrm {PZF}} _{l,k,t} + \tilde { \mathbf {h}}_{l,k} ^{\text {H}} \mathbf {w} ^ {\textrm{PZF}} _{l,{t}})}\right |^{2}}\right \},\notag\\&=\left ({ \sum \limits _{l \in \mathcal {Z} _{k}}\sqrt {\rho _{l,k}}\alpha ^ {\textrm{PZF}} _{l,k,k}}\right )^{2} \!\!+ \sum \limits _{t=1}^{K} { \sum \limits _{l \in \mathcal {Z} _{t}}}\rho _{l,t}(\beta _{l,k}\!-\!\gamma _{l,k}),\notag
\\
&=\left ({ \sum \limits _{l \in \mathcal {Z} _{k}}\!\sqrt {\rho _{l,k}(M\!-\! \tausl)\gamma _{l,k}}}\right)^{2} \! +\! \sum \limits_{t=1}^{K}{ \sum \limits _{l \in \mathcal {Z} _{t}}}\rho _{l,t}(\beta _{l,k}\!-\!\gamma _{l,k}).
\end{align}
%-----------------

If $k \notin \mathcal{S}_{l}$, then $k \in \mathcal{W}_{l}$ and $\mathbf{h}_{l, k}$ is independent of $\mathbf{w}_{l, t}^{\textrm{PZF}}$ $\forall t \neq{k}$. Hence,
%------------------
\begin{align}\label{eq:ExBUK4}&\sum \limits _{t=1}^{K} \mathop {\mathrm {\Ex }}\nolimits \left \{{\left |{ \sum \limits _{l \in \mathcal {Z} _{t}}\sqrt {\rho_{l,t}} \mathbf {h}_{l,k} ^{\text {H}} \mathbf {w}^{\textrm{PZF}} _{l,{t}}}\right |^{2}}\right \}\notag\\
&=\sum \limits _{t=1}^{K}\sum \limits _{l \in \mathcal{Z} _{t}}\rho _{l,t} \mathop {\mathrm {\Ex }}\nolimits \left \{{\left |{ \mathbf{h}_{l,k} ^{\text{H}} \mathbf{w} ^ {\textrm{PZF}} _{l,{t}}}\right |^{2}}\right \}, \notag\\
&=\sum \limits _{t=1}^{K}\sum_{l \in \mathcal{Z}_{t}} \rho_{l,t} \Ex \left\{\mathbf{h}_{l, k}^{\text{H}} \Ex\left\{\mathbf{w}_{l, t}^{\mathrm{PZF}}\left(\mathbf{w}_{l, t}^{\mathrm{PZF}}\right)^{\text{H}}\right\} \mathbf{h}_{l, k}\right\},\notag\\
&=\sum \limits _{t=1}^{K}\sum_{l \in \mathcal{Z}_{t}} \rho_{l, t}\beta_{l,k}.
\end{align}
%-------------
Now, we combine \eqref{eq:ExBUK3} and \eqref{eq:ExBUK4} to obtain
\begin{align}\label{eq:ExBUK5} &\sum \limits _{t=1}^{K} \mathop {\mathrm {\Ex }}\nolimits \left \{{\left |{ \sum \limits _{l \in \mathcal {Z} _{t}}\sqrt {\rho_{l,t}} \mathbf{h}_{l,k} ^{\text{H}} \mathbf {w}^{\textrm{PZF}} _{l,{t}}}\right |^{2}}\right \}\notag\\&=\left ({ \sum \limits _{l \in \mathcal {Z} _{k}}\!\sqrt {\rho _{l,k}(M\!-\! \tausl)\gamma _{l,k}}}\right)^{2} \!\!\! +\! \sum \limits_{t=1}^{K}{ \sum \limits _{l \in \mathcal {Z} _{t}}}\rho _{l,t}(\beta _{l,k}\!-\!\delta_{l,k}\gamma _{l,k}).\end{align}
Similarly, the third term of the right-hand-side in \eqref{eq:ExBUK2} can be computed as follows:

If $k \in \mathcal{W}_{p}$, then $\mathbf{h}_{p, k}$ is independent of $\mathbf{w}_{p, t}^{\textrm{PMRT}}$ $\forall t \neq{k}$. Thus, we obtain
%---------------------
\begin{align}
\label{eq:ExBUK6}
&\sum \limits _{t=1}^{K} \mathop {\mathrm {\Ex }}\nolimits \left \{{\left |{ \sum \limits _{p \in \mathcal {M} _{t}}\sqrt {\rho_{p,t}} \mathbf {h}_{p,k} ^{\text {H}} \mathbf {w}^{\textrm{PMRT}} _{p,{t}}}\right |^{2}}\right \}\notag\\&=\sum \limits _{t=1}^{K}\sum \limits _{p \in \mathcal {M} _{t}}\rho _{p,t} \mathop {\mathrm {\Ex }}\nolimits \left \{{\left |{ \mathbf {h}_{p,k} ^{\text{H}} \mathbf {w} ^ {\textrm{PMRT}} _{p,t}}\right |^{2}}\right \} \notag\\&=\left ({ \sum \limits _{l \in \mathcal {Z} _{k}}\!\sqrt {\rho _{l,k}(M\!-\! \tausl)\gamma _{l,k}}}\right)^{2}\notag\\&+\sum \limits _{t=1}^{K}\sum_{p \in \mathcal{M}_{t}} \rho_{p,t} \Ex\left\{\mathbf{h}_{p, k}^{\mathrm{H}} \Ex\left\{\mathbf{w}_{p, t}^{\mathrm{PMRT}}\left(\mathbf{w}_{p, t}^{\mathrm{PMRT}}\right)^{\text{H}}\right\} \mathbf{h}_{p, k}\right\}\notag\\&=\left ({ \sum \limits _{p \in \mathcal {M} _{k}}\!\sqrt {\rho _{l,k}(M\!-\! \tausp)\gamma _{l,k}}}\right)^{2}+\sum \limits _{t=1}^{K}\sum_{p \in \mathcal{M}_{t}} \rho_{p, t}\beta_{p,k}.
\end{align}
%-------------

If $k \in \mathcal{S}_{p}$, then $\hat { \mathbf {h}}_{p,k}\mathbf{w}_{p, t}^{\textrm{PMRT}}=0$, and we have
%----------
\begin{align}\label{eq:ExBUK7}
&\sum \limits _{t=1}^{K} \mathop {\mathrm {\Ex }}\nolimits \left \{{\left |{ \sum \limits _{p \in \mathcal {M} _{t}}\sqrt {\rho_{p,t}} \mathbf {h}_{p,k} ^{\text {H}} \mathbf {w}^{\textrm{PMRT}} _{p,t}}\right |^{2}}\right \}\notag\\&=\sum \limits _{t=1}^{K} \mathop {\mathrm {\Ex }}\nolimits \left \{{\left |{ \sum \limits _{p \in \mathcal {M} _{t}}\sqrt {\rho _{p,t}}\tilde { \mathbf {h}}_{p,k} ^{\text {H}} \mathbf {w} ^ {\textrm{PMRT}} _{p,t}}\right |^{2}}\right \},\notag\\&=\sum \limits _{t=1}^{K} { \sum \limits _{p \in \mathcal {M} _{t}}}\rho _{p,t}(\beta _{p,k}\!-\!\gamma _{p,k}).\end{align}
Adding \eqref{eq:ExBUK6} to \eqref{eq:ExBUK7}, we obtain
\begin{align}\label{eq:ExBUK8}&\sum \limits _{t=1}^{K} \mathop {\mathrm {\Ex }}\nolimits \left \{{\left |{ \sum \limits _{ \in \mathcal {M} _{t}}\sqrt {\rho_{p,t}} \mathbf {h}_{p,k} ^{\text {H}} \mathbf {w}^{\textrm{PMRT}} _{p,{t}}}\right |^{2}}\right \}\notag\\&=\!\!\left ({ \sum \limits _{p \in \mathcal {M} _{k}}\!\!\!\sqrt {\rho _{p,k}(M\!-\! \tausp)\gamma _{p,k}}}\right)^{2}\!\!\!\!\!+\!\!\!\sum \limits _{t=1}^{K} \!{ \sum \limits_{p \in \mathcal {M} _{t}}}\!\!\rho _{p,t}(\beta _{p,k}\!\!-\!\delta_{p,k}\gamma _{p,k}).
\end{align}
%==============
The second term of \eqref{eq:ExBUK2} is equal to 0 $\forall t \neq{k}$. Moreover, for $t=k$, we obtain
%---------------
\begin{align}\label{eq:ExBUK9}
\hspace {-1.8pc} &2 \sum_{t=1}^{K}\! \operatorname{Re}\left\{\sum_{l \in \mathcal{Z}_{k}} \!\sum_{p \in \mathcal{M}_{k}}\!\! \sqrt{\rho_{l, k} \rho_{p, k}} \Ex \left\{\mathbf{h}_{l, k}^{\mathrm{H}} \mathbf{w}_{l,k}^{\textrm{PZF}}\left(\mathbf{w}_{p,k}^{\textrm{PMRT}}\right)^{\mathrm{H}}\! \mathbf{h}_{p, k}\right\}\right\}\notag\\&=\!2\!\sum_{l \in \mathcal{Z}_{k}}\!\sum_{p \in \mathcal{M}_{k}} \!\!\sqrt{\rho_{l,k} \rho_{p,k}} \Ex \left\{\mathbf{h}_{l, k}^{\mathrm{H}} \mathbf{w}_{l,k}^{\textrm{PZF}}\right\}\Ex \left\{\mathbf{h}_{p, k}^{\mathrm{H}}\mathbf{w}_{p,k}^{\textrm{PMRT}}\right\},\notag \\&=\!2\!\sum_{l \in \mathcal{Z}_{k}}\!\sum_{p \in \mathcal{M}_{k}} \!\!\sqrt{\rho_{l,k}\rho_{p,k}(M \!\!- \!\!\tausl)(M\!\!-\!\!\tausp)\gamma_{l,k}\gamma_{p,k}}.
\end{align}
%-----------------
Now, using  \eqref{eq:ExBUK5}, \eqref{eq:ExBUK8}, and \eqref{eq:ExBUK9}, we can rewrite \eqref{eq:ExBUK2} as
%----------------
\begin{align}\label{eq:ExBUK10}
&\sum_{t=1}^{K}\! \Ex \!\left\{\left|\sum_{l \in \mathcal{Z}_{t}} \!\!\sqrt{\rho_{l, t}} \mathbf{h}_{l, k}^{\mathrm{H}} \mathbf{w}_{l, {t}}^{\textrm{PZF}}\!+\!\sum_{p \in \mathcal{M}_{t}} \!\!\sqrt{\rho_{p, l}} \mathbf{h}_{p, k}^{\mathrm{H}} \mathbf{w}_{p, {t}}^{\textrm{PMRT}}\right|^{2}\right\}\!\notag \\&=\left ({ \sum \limits _{l \in \mathcal {Z} _{k}}\!\sqrt {\rho _{l,k}(M\!-\! \tausl)\gamma _{l,k}}}\right)^{2} \! +\! \sum \limits_{t=1}^{K}{ \sum \limits _{l \in \mathcal {Z} _{t}}}\rho _{l,t}(\beta _{l,k}\!-\!\delta_{l,k}\gamma _{l,k})\notag\\&+\!\!\left ({ \sum \limits _{p \in \mathcal {M} _{k}}\!\!\!\sqrt {\rho _{p,k}(M\!-\! \tausp)\gamma _{p,k}}}\right)^{2}\!\!\!\!\!+\!\!\sum \limits _{t=1}^{K} \!{ \sum \limits _{p \in \mathcal {M} _{t}}}\!\!\rho _{p,t}(\beta _{p,k}\!-\!\delta_{p,k}\gamma _{p,k})\notag\\&+\!2\!\sum_{l \in \mathcal{Z}_{k}}\!\sum_{p \in \mathcal{M}_{k}} \!\!\sqrt{\rho_{l,k}\rho_{p,k}(M \!\!- \!\!\tausl)(M\!\!-\!\!\tausp)\gamma_{l,k}\gamma_{p,k}},\notag\\&=\left ({ \sum \limits _{l \in \mathcal {Z} _{k}}\!\sqrt {\rho _{l,k}(M\!-\! \tausl)\gamma _{l,k}}}+ { \sum \limits _{p \in \mathcal {M} _{k}}\!\sqrt {\rho _{p,k}(M\!-\! \tausp)\gamma _{p,k}}}\right)^{2} \!\!\! \notag\\&+\sum \limits_{t=1}^{K}\sum \limits_{l=1}^{L}\!\!\rho _{l,t}(\beta _{l,k}\!-\!\delta_{l,k}\gamma _{l,k}),\notag\\&=\left ({ \sum \limits_{l=1}^{L}\!\sqrt {\rho _{l,k}(M\!-\! \tausl)\gamma _{l,k}}}\right)^{2} \! +\! \sum \limits_{t=1}^{K}\sum \limits_{l=1}^{L}\rho _{l,t}(\beta _{l,k}\!-\!\delta_{l,k}\gamma _{l,k}).
\end{align}
%---------------
Substituting \eqref{eq:ExBUK} and \eqref{eq:ExBUK10} in \eqref{eq:SINRk}, $\textrm {SINR}_{k}^{\textrm{PPZF}}$ in \eqref{eq:SINR_k2} is obtained.
We next provide a closed-form expression for the SE for the eavesdropper.

%=====================================
\section{Proof of Proposition~\ref{prop:SINRE}}~\label{app:SINRE} 
%%%%%%%%%%%%%%%%%%%%%%%%%%%%%%%%%%%%%%%%%%%
Based on~\eqref{I_E},  the received SINR at the eavesdropper can be written as
%=
  \begin{align}~\label{eq:SINRE}
\SINRE= \frac {\Ex\left\{ |\BUE|^{2} \right\}}{{\sum\limits_{\substack{ t \neq 1}}^{K} }{{ \Ex\left\{ |\mathrm {UI}_{ \textrm {E}, t}|^{2}\right\}}} + 1 }.
\end{align} 
%==================
We first calculate $\Ex \left \{{{|\BUE|^{2}}}\right \}$. By denoting $\hat { \qh}_{p,E}=\sqrt{\alpha_{p, 1}} \qh_{p,1}$, where $\alpha_{p,1}=\big(\rho_{E} \beta_{p,E}^{2}\big) /\big(\rho_{u} \beta_{p,1}^{2}\big)$, and $\gamma_{p,E}=\alpha_{p,1} \gamma_{p,1}$, we have
%===========================
\begin{align}~\label{eq:ExBUE}
	&\Ex \left \{{{|\BUE|^{2}}}\right \}\notag\\
		&=\Ex\Bigg\{\bigg|\sum_{l \in \Zone} \sqrt{\rho_{l, 1}} \qh_{l, E}^{\mathrm{H}} \qw_{l, 1}^{\PZF}+\sum_{p \in \Mone} \sqrt{\rho_{p, 1}} \qh_{p, E}^{\mathrm{H}} \qw_{p, 1}^{\PMRT}\bigg|^{2}\Bigg\}
  \notag\\
		&=\Ex\Bigg\{\!\bigg|\sum_{l \in \Zone}
  \!\!\sqrt{\rho_{l, 1}} \qh_{l, E}^{\mathrm{H}} \qw_{l, 1}^{\PZF}\bigg|^{2}\Bigg\}\!\!+\!
  \Ex\Bigg\{\!\bigg|\!\sum_{p \in \Mone}\!\!\!\! \sqrt{\rho_{p, 1}} \qh_{p, E}^{\mathrm{H}} \qw_{p, 1}^{\PMRT}\bigg|^{2}\!\Bigg\}
  \notag\\
		&+ 2\! \operatorname{Re}\Bigg\{\!\sum_{l \in \Zone} \sum_{p \in \Mone}\!\! \sqrt{\rho_{l, 1} \rho_{p, 1}} \Ex\left\{\qh_{l, E}^{\mathrm{H}} \qw_{l, 1}^{\PZF}\!\left(\qw_{p, 1}^{\PMRT}\!\right)^{\mathrm{H}}\! \qh_{p, E}\!\right\}\!\!\Bigg\},
  \notag\\
		&=\left(\sum_{l \in \Zone} \sqrt {{\rho _{l,1}}(M-\tausl)\gamma_{l,E}}\right)^{2}+\sum_{l \in \Zone} \rho_{l, 1}(\beta_{l,E}-\gamma_{l,E})\notag\\&+\sum_{p \in \Mone} \rho_{p, 1} \Ex\left\{\left|\hat { \qh}_{p,E} ^{\text {H}}\qw_{p, 1}^{\PMRT}\right|^{2}\right\}
  \notag
  \\
  &+\sum_{p \in \Mone}\!\sum_{q \in \Mone, q \neq p}\!\!\!\!\!\!\!\!
  \sqrt{\rho_{p, 1}\rho_{q, 1}}\Ex\left\{\hat { \qh}_{p,E} ^{\text {H}}\qw_{p, 1}^{\PMRT}\right\}\Ex\left\{(\qw_{q, 1}^{\PMRT})^{\text{H}}\hat { \qh}_{q,E}\right\}
  \notag\\
  &+\sum_{p \in \Mone} \rho_{p, 1}(\beta_{p,E}-\gamapE)\notag\\&+2\sum_{l \in \Zone}\sum_{p \in \Mone}\sqrt{\rho_{l, 1}\rho_{p, 1}(M -\tausl)(M\!\!-\!\!\tausp)\gamma_{l,E}\gamapE} ,
\end{align}
where we have used $\gamalmE=\alpha_{l,1} \gamma_{l,1}$, and the fact that $\tilde { \qh}_{p,E}= { \qh}_{p,E}-\hat { \qh}_{p,E}$ is independent of $\hat { \qh}_{p,E}$ and has zero mean. 
Since $\hat { \mathbf {h}}_{p,E} ^{\text {H}}=\sqrt{\alpha_{p,1}}\hat { \mathbf {h}}_{p,1} ^{\text {H}}$, the third term of \eqref{eq:ExBUE} can be re-written as
\begin{align}~\label{eq:ExBUE2}
&\sum_{p \in \mathcal{M}_{1}} \rho_{p, 1} \Ex \left\{\left|\hat { \mathbf {h}}_{p,E} ^{\text {H}}\mathbf{w}_{p, 1}^{\textrm{PMRT}}\right|^{2}\right\}\notag\\&=\sum_{p \in \mathcal{M}_{1}}  \frac{\rho_{p, 1}\alpha_{p,1}}{(M-\tausp)\gamma_{p,1}} \Ex \left\{\left|\hat {\mathbf {h}}_{p,1} ^{\text {H}} \mathbf{B}_{p}\hat { \mathbf {h}}_{p,1}\right|^{2}\right\},\notag\\&=\sum_{p \in \mathcal{M}_{1}}  \rho_{p, 1}\gamma_{p,E}(M-\tausp+1).
\end{align}
%===========================test
Moreover, by using \eqref{eq:wPMRT} and the fact that $\Ex\left\{\hat {\qh}_{p,1} ^{\text {H}}\qB_{p}\hat { \qh}_{p,1}\right\}=\Ex\left\{\trace(\mathbf{B}_{p}\hat { \qh}_{p,1}\hat {\qh}_{p,1} ^{\text {H}})\right\}$, we have
\begin{align}~\label{eq:ExBUE3}
&\hspace {0pc} \sum_{p \in \mathcal{M}_{1}}\sqrt{\rho_{p, 1}}\Ex \left\{\hat { \mathbf {h}}_{p,E}^{\text {H}}\mathbf{w}_{p, 1}^{\textrm{PMRT}}\right\}\notag\\&=\sum_{p \in \mathcal{M}_{1}} \frac{\sqrt{\rho_{p, 1}\alpha_{p,1}}}{\sqrt{(M-\tausp)\gamma_{p,1}}} \Ex \left\{\hat {\mathbf {h}}_{p,1} ^{\text {H}} \mathbf{B}_{p}\hat { \mathbf {h}}_{p,1}\right\},\notag\\&=\sum_{p \in \mathcal{M}_{1}} \frac{\sqrt{\rho_{p, 1}\alpha_{p,1}}}{\sqrt{(M-\tausp)\gamma_{p,1}}} \Ex \left\{\trace(\mathbf{B}_{p}\hat { \mathbf {h}}_{p,1}\hat {\mathbf {h}}_{p,1} ^{\text {H}})\right\},\notag\\&=\sum_{p \in \mathcal{M}_{1}} \frac{\sqrt{\rho_{p, 1}\alpha_{p,1}}}{\sqrt{(M-\tausp)\gamma_{p,1}}} \gamma_{p,1}\Ex \left\{\trace(\mathbf{B}_{p})\right\},\notag\\&=\sum_{p \in \mathcal{M}_{1}} \sqrt{\rho_{p, 1}(M-\tausp)\gamma_{p,E}}.
\end{align}
%===========================test
Thus, we can rewrite the fourth term of \eqref{eq:ExBUE} as
\begin{align}~\label{eq:ExBUE4}
&\hspace {0pc}\sum_{p \in \mathcal{M}_{1}}\!\sum_{q \in \mathcal{M}_{1}}\sqrt{\rho_{p, 1}\rho_{q, 1}}\Ex \left\{\hat { \mathbf {h}}_{p,E} ^{\text {H}}\mathbf{w}_{p, 1}^{\textrm{PMRT}}\right\}\!\Ex \left\{(\mathbf{w}_{q, 1}^{\textrm{PMRT}})^{\text{H}}\hat { \mathbf {h}}_{q,E}\right\}\notag\\&=\sum_{p \in \mathcal{M}_{1}}\sum_{q \in \mathcal{M}_{1}} \sqrt{\rho_{p, 1}\rho_{q, 1}(M - \tausp)(M -\tausq)\gamma_{p,E}\gamma_{q,E}}.
\end{align}
%===========================
To this end, by substituting \eqref{eq:ExBUE2} and \eqref{eq:ExBUE4} into \eqref{eq:ExBUE}, we obtain
%===========================
\begin{align}~\label{eq:ExBUE5}
&\Ex \left \{{{|\BUE|^{2}}}\right \}=\Bigg(\sum_{l=1}^{L} \sqrt {{\rho _{l,1}}(M\!-\!\tausl)\gamma_{l,E}}\Bigg)^{2}\nonumber\\
&+\sum_{l \in \Zone} \rho_{l, 1}(\beta_{l,E}\!-\!\gamma_{l,E})
\!+\!\sum_{p \in \Mone}\! \rho_{p, 1}\beta_{p,E},\notag
\\&=\left(\sum_{l=1}^{L} \sqrt {{\rho _{l,1}}(M-\tausl)\gamma_{l,E}}\right)^{2}\!\!\!+\!\!\sum_{l=1}^{L}\rho_{l, 1}\beta_{l,E}\!-\!\!\sum_{l \in \Zone}\!\! \rho_{l, 1}\gamma_{l,E}.
  \end{align}
%===========================
Similarly, for $t\neq{1}$, we calculate $\Ex \left \{|\UIEt|^{2}\right \}$ as
%===========================
\begin{align}~\label{eq:ExUIE2}
&\hspace{-1em}\Ex \big \{|\UIEt|^{2}\big\}\!
  =\!\Ex\Bigg\{\!\Big|\!\sum_{l \in \Zt} \!\!\!\sqrt{\rholt} \qh_{l, E}^{\mathrm{H}} \qw_{l, t}^{\PZF}\!\!+\!\!\!\!\sum_{p \in \Mt} \!\!\!\!\!\sqrt{\rho_{p,t}} \qh_{p, E}^{\mathrm{H}} \qw_{p,t}^{\PMRT}\Big|^{2}\Bigg\},
\notag\\
&\hspace{-1.5em}=
\Ex\bigg\{\bigg|\sum_{l \in \Zt} \sqrt{\rho_{l,t}} \qh_{l, E}^{\mathrm{H}} \qw_{l,t}^{\PZF}\bigg|^{2}\bigg\}\!+\!\Ex\bigg\{\!\bigg|\sum_{p \in \Mt}\!\! \sqrt{\rho_{p,t}} \qh_{p, E}^{\mathrm{H}} \qw_{p,t}^{\PMRT}\bigg|^{2}\!\bigg\}
\notag\\
+& 
2 \operatorname{Re}\!\Bigg\{\!\!\sum_{l \in \Zt}\! \sum_{p \in \Mt} \!\!\!\!\sqrt{\rho_{l,t} \rho_{p,t}} \Ex\left\{\!\qh_{l, E}^{\!\mathrm{H}} \qw_{l,t}^{\PZF}
\big(\qw_{p,t}^{\PMRT}\big)^{\!\mathrm{H}}
\!\qh_{p, E}\!\right\}\!\!\Bigg\}\!.
\end{align}
%===========================test

%===========================test
The first term of \eqref{eq:ExUIE2} is computed for two conditions: 1) when user $1$ $\in \Sl$, then $\hat{\qh}_{l, 1}^{\mathrm{H}} \qw_{l, t}^{\PZF}=0$ and  $\hat{\qh}_{l, E}^{\mathrm{H}} \qw_{l, t}^{\PZF}=0$ ($\hat{\qh}_{l, E}=\sqrt{\alpha_{l,1}}\hat{\qh}_{l, 1}$); 2) when user $1$ $\notin \Sl$ and hence $\hat{\qh}_{l, E}$ is independent of $\qw_{l, t}^{\PZF}$. Then, we obtain
%===========================
\begin{align}~\label{eq:ExBUE3}
		\Ex\bigg\{\bigg|\sum_{l \in \Zt} \!\!\sqrt{\rho_{l,t}} \qh_{l, E}^{\mathrm{H}} \qw_{l,t}^{\PZF}\bigg|^{2}\bigg\}\!=\!\sum_{l \in \Zt}\! \rholt(\beta_{l,E}-\delta_{l,1}\gamma_{l,E}).
\end{align}
%===========================
Similarly, for the the second term of \eqref{eq:ExUIE2}, $\hat{\qh}_{p, 1}^{\mathrm{H}}\qB_{p}=0$ and $\hat{\qh}_{p, E}^{\mathrm{H}}\mathbf{B}_{p}=0$ when user $1$ $\in\mathcal{S}_{p}$. Moreover, $\qh_{p, 1}$ is independent of $\qw_{p, t}^{\PMRT}$ when user $1$ $\notin\mathcal{S}_{p}$. Then, we obtain
%===========================
\vspace{-0.2em}
\begin{align}~\label{eq:ExPMRT}
		&\hspace {0pc} 
		\Ex\bigg\{\!\bigg|\!\sum_{p \in \Mt} \!\!\!\!\sqrt{\rho_{p,t}}\qh_{p, E}^{\mathrm{H}} \qw_{p,t}^{\PMRT}\bigg|^{2}\bigg\}\!\!
		=\!\!\!\!\sum_{p \in \Mt} \!\!\!\rho_{p,t}(\beta_{p,E}\!-\!\delta_{p,1}\gamapE).
\end{align}
%===========================
The third term of \eqref{eq:ExUIE2} is equal to 0 as $\qh_{l, E}$ and $\qh_{p, E}$ are independent of both $\qw_{l, t}^{\PZF}$ and $\qw_{p, t}^{\PMRT}$, respectively. By substituting~\eqref{eq:ExBUE3} and~\eqref{eq:ExPMRT} into~\eqref{eq:ExUIE2}, we obtain
%===========================
\begin{align}~\label{eq:ExUIE3}
		\Ex \left \{|\mathrm \UIEt|^{2}\right \}&=\sum_{l \in \Zt} \rholt(\beta_{l,E}-\delta_{l,1}\gamma_{l,E})\notag\\
		&+\sum_{p \in \Mt} \rho_{p,t}(\beta_{p,E}-\delta_{p,1}\gamapE),\notag\\
		&=\sum_{l=1}^{L} \rho_{p,t}(\betalE-\delta_{l,1}\gamalmE).
\end{align}
%===========================
By plugging \eqref{eq:ExBUE2} and \eqref{eq:ExUIE3} into \eqref{eq:SINRE}, $\SINRE$ in \eqref{eq:SINR_E2} can be obtained.
%=====================================

\vspace{0em}
\balance
\bibliographystyle{IEEEtran}
\bibliography{bibliography.bib}
\end{document}